%% file: modeling.tex
\documentclass{amsart}
\usepackage{amssymb}
\usepackage{amsmath}
\usepackage{amsthm}
\usepackage{mathrsfs}
\usepackage{colortbl}
\newtheorem{theorem}{Theorem}
\newtheorem{lemma}[theorem]{Lemma}

\newtheorem{remark}[theorem]{Remark}

\theoremstyle{definition}
\newtheorem{definition}[theorem]{Definition}

\newtheorem{proposition}[theorem]{Proposition}

\makeatletter

\@addtoreset{theorem}{section}
\makeatother

\makeatletter

\@addtoreset{equation}{section}
\makeatother

\begin{document}               

\title[Inviscid multiphase flow system with surface flow and tension]{Energetic variational approaches for inviscid multiphase flow systems with surface flow and tension}
\author[Hajime Koba]{Hajime Koba}                                
\address{Graduate School of Engineering Science, Osaka University,\\
1-3 Machikaneyamacho, Toyonaka, Osaka, 560-8531, Japan}                                  
%\curraddr{...}                                   
\email{iti@sigmath.es.osaka-u.ac.jp}

%\date{}                                      
%\thanks{This work was partly supported by the Japan Society for the Promotion of Science (JSPS) KAKENHI Grant Number JP21K03326.}                                     
%\translator{...}                                 
\begin{abstract}
We consider the governing equations for the motion of the inviscid fluids in two moving domains and an evolving surface from an energetic point of view. We employ our energetic variational approaches to derive inviscid multiphase flow systems with surface flow and tension. More precisely, we calculate the variation of the flow maps to the action integral for our model to derive both surface flow and tension. We also study the conservation and energy laws of our multiphase flow systems. The key idea of deriving the pressure of the compressible fluid on the surface is to make use of the feature of the barotropic fluid, and the key idea of deriving the pressure of the incompressible fluid on the surface is to apply a generalized Helmholtz-Weyl decomposition on a closed surface.
\end{abstract}       
\maketitle

\section{Introduction}\label{sect1}

\begin{figure}[htbp]
\input{pic1moto.tex}
\caption{Moving Domains and Surfaces}
\label{Fig1}
\end{figure}
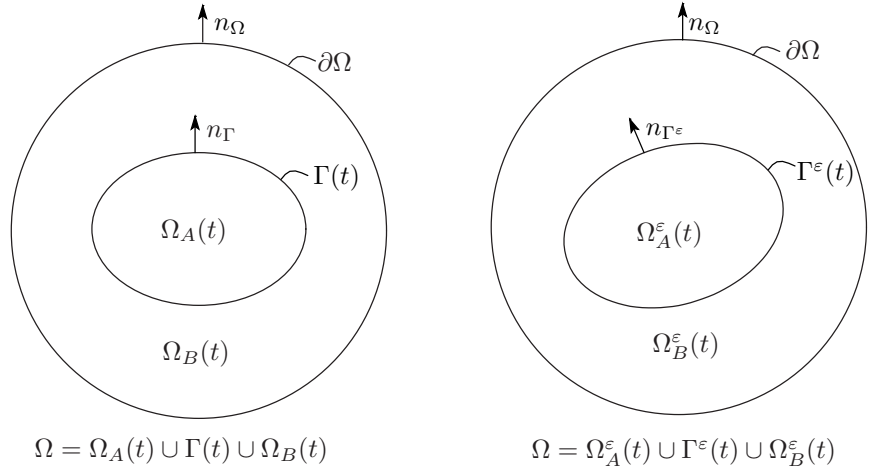

We are interested in a mathematical modeling of a soap bubble floating in the air and an air bubble moving in water. When we focus on a soap bubble, we can see the fluid flow in the bubble. We call the fluid flow in the bubble a \emph{surface flow}. We can consider a surface flow as fluid flow on an evolving surface. In order to make a mathematical model of a soap bubble floating in the air, we have to study the dependencies among fluid-flows in domains and surface flow. This paper considers the governing equations for the motion of the inviscid fluids in two moving domains and an evolving surface from an energetic point of view. We employ our energetic variational approaches to derive three inviscid multiphase flow systems with surface flow and tension.

Let us first introduce fundamental notations. Let $t \geq 0$ be the time variable, $x( = { }^t (x_1 , x_2, x_3 ) )$, $\xi_A ( = { }^t (\xi^A_1 , \xi^A_2, \xi^A_3 ) )$, $\xi_B ( = { }^t (\xi^B_1 , \xi^B_2, \xi^B_3 ) )$, $\xi_S ( = { }^t (\xi^S_1 , \xi^S_2, \xi^S_3 ) )$ $\in \mathbb{R}^3$ the spatial variables, and $X(= { }^t (X_1,X_2)) \in \mathbb{R}^2$ the spatial variable. Fix $T >0$. Let $\Omega \subset \mathbb{R}^3$ be a bounded domain with a smooth boundary $\partial \Omega$. The symbol $n_\Omega = n_\Omega (x) = { }^t (n^\Omega_1 , n^\Omega_2 , n^\Omega_3 )$ denotes the unit outer normal vector at $x \in \partial \Omega$. Let $\Omega_A (t) (= \{ \Omega_A (t) \}_{0 \leq t < T} )$ be a bounded domain in $\mathbb{R}^3$ with a moving boundary $\Gamma (t)$. Assume that $\Gamma (t) (= \{ \Gamma (t) \}_{0 \leq t < T})$ is a smoothly evolving surface. The symbol $n_\Gamma = n_\Gamma ( x , t ) = { }^t (n^\Gamma_1 , n^\Gamma_2 , n^\Gamma_3)$ denotes the unit outer normal vector at $x \in \Gamma (t)$. For each $t \in [0,T)$, assume that $\Omega_A (t) \Subset \Omega$. Set $\Omega_B (t) = \Omega \setminus \overline{\Omega_A (t)}$. It is clear that $\Omega = \Omega_A (t) \cup \Gamma (t) \cup \Omega_B (t)$ (see Figure \ref{Fig1}). Set
\begin{multline}\label{eq11}
\Omega_{A,T} = \bigcup_{0< t < T} \{ \Omega_A (t) \times \{ t \} \},{ \ } \Omega_{B,T} = \bigcup_{0< t < T} \{ \Omega_B (t) \times \{ t \} \},\\
\Gamma_T = \bigcup_{0< t < T} \{ \Gamma (t) \times \{ t \} \},{ \ }\Omega_T = \Omega \times (0,T),{ \ }\partial \Omega_T = \partial \Omega \times (0,T).
\end{multline}

In this paper we assume that the fluids in $\Omega_{A,T}$ and $\Omega_{B,T}$ are barotropic compressible ones, and that the fluid on $\Gamma_T$ is compressible or incompressible one. An incompressible fluid on the surface means one that satisfies the surface divergence-free condition. Let us state physical notations. Let $\rho_A = \rho_A ( x , t)$, $v_A = v_A ( x , t) = { }^t (v^A_1 , v^A_2 , v^A_3 )$, and $\mathfrak{p}_A = \mathfrak{p}_A (x,t)$ be the density, the velocity, and the pressure of the fluid in $\Omega_A (t)$, respectively. Let $\rho_B = \rho_B ( x , t)$, $v_B = v_B ( x , t) = { }^t (v^B_1 , v^B_2 , v^B_3 )$, and $\mathfrak{p}_B = \mathfrak{p}_B (x,t)$ be the density, the velocity, and the pressure of the fluid in $\Omega_B (t)$, respectively. Let $\rho_S = \rho_S ( x , t)$, $v_S = v_S ( x , t) = { }^t (v^S_1 , v^S_2 , v^S_3 )$, and $\mathfrak{p}_S = \mathfrak{p}_S (x,t)$, $\Pi_S = \Pi_S (x,t)$ be the density, the velocity, and the pressures of the fluid on $\Gamma (t)$, respectively. 

\begin{remark}
\noindent $(\rm{i})$ The symbol $\mathfrak{p}_S$ corresponds to the pressure of the compressible fluid on $\Gamma_T$, and $\Pi_S$ corresponds to the pressure of the incompressible fluid on $\Gamma_T$. We call $\mathfrak{p}_S$, $\Pi_S$ \emph{total pressures}. Total pressure means one that includes surface pressure and tension.\\
\noindent $(\rm{ii})$ We call $v_S$ a \emph{total velocity} on $\Gamma_T$. Total velocity means that $v_S$ can be divided into surface velocity $u_S$ and motion velocity $w_S$, that is, $v_S = u_S + w_S$. The surface flow $u_S$ is a tangential vector on $\Gamma (t)$. The motion velocity $w_S$ is the speed of the evolving surface $\Gamma (t)$. Therefore, the total velocity $v_S$ is not a necessary tangential vector on $\Gamma (t)$. In this paper, we assume that $w_S$ is a normal vector on $\Gamma (t)$, and focus on the total velocity and total pressure to make our models.
\end{remark}

Let us explain the basic assumptions of mathematical modeling of inviscid multiphase flow systems with surface flow and tension. We assume that
\begin{equation}\label{eq12}
\begin{cases}
v_B \cdot n_\Omega =0 & \text{ on }\partial \Omega_T,\\
v_A \cdot n_\Gamma = v_B \cdot n_\Gamma = v_S \cdot n_\Gamma & \text{ on } \Gamma_T.
\end{cases}
\end{equation}
The condition $v_B \cdot n_\Omega =0$ means that fluid particles do not go out of the domain $\Omega$.

This paper has two purposes. The first one is to apply our energetic variational approaches to derive three multiphase flow systems. The first system is the following inviscid multiphase flow system with compressible surface flow:
\begin{equation}\label{eq13}
\begin{cases}
D_t^A \rho_A + ({\rm{div}} v_A) \rho_A = 0 & \text{ in } \Omega_{A,T},\\
\rho_A D_t^A v_A + {\rm{grad}} \mathfrak{p}_A = 0 & \text{ in } \Omega_{A,T},\\
D_t^B \rho_B + ({\rm{div}} v_B) \rho_B = 0 & \text{ in } \Omega_{B,T},\\
\rho_B D_t^B v_B + {\rm{grad}} \mathfrak{p}_B = 0 & \text{ in } \Omega_{B,T},\\
D_t^S \rho_S + ({\rm{div}}_\Gamma v_S) \rho_S = 0 & \text{ on } \Gamma_{T},\\
\rho_S D_t^S v_S + {\rm{grad}}_\Gamma \mathfrak{p}_S + \mathfrak{p}_S H_\Gamma n_\Gamma + \mathfrak{p}_B n_\Gamma - \mathfrak{p}_A n_\Gamma = 0 & \text{ on } \Gamma_{T},
\end{cases}
\end{equation}
with \eqref{eq12} and
\begin{equation}\label{eq14}
\begin{cases}
\mathfrak{p}_A = \mathfrak{p}_A ( \rho_A) =  \rho_A p_A' (\rho_A) - p_A (\rho_A) & \text{ in } \Omega_{A,T},\\
\mathfrak{p}_B = \mathfrak{p}_B ( \rho_B) =  \rho_B p_B' (\rho_B) - p_B (\rho_B ) & \text{ in } \Omega_{B,T},\\
\mathfrak{p}_S = \mathfrak{p}_S ( \rho_S) =  \rho_S p_S' (\rho_S) - p_S (\rho_S ) & \text{ on } \Gamma_{T}.
\end{cases}
\end{equation}
\noindent Here $p_A$, $p_B$, $p_S$ are three $C^1$-functions, $p' = p'(r) = {dp}/{dr}(r)$, $D_t^A f = \partial_t f + (v_A,\nabla ) f $, $D_t^B f = \partial_t f + (v_B, \nabla ) f $, $D_t^S f = \partial_t f + (v_S, \nabla ) f $, $(v_A , \nabla ) f = v^A_1 \partial_1 f + v^A_2 \partial_2 f + v^A_3 \partial_3 f$, ${\rm{div}}v_A = \nabla \cdot v_A$, ${\rm{grad}} f = \nabla f$, $\nabla = { }^t (\partial_1 , \partial_2 , \partial_3)$, $\partial_i = \partial/{\partial x_i}$, $\partial_t = \partial/{\partial t}$, ${\rm{div}}_\Gamma v_S = \nabla_\Gamma \cdot v_S$, ${\rm{grad}}_\Gamma f = \nabla_\Gamma f$, $\nabla_\Gamma = { }^t (\partial^\Gamma_1 , \partial^\Gamma_2 , \partial^\Gamma_3 ) $, $\partial^\Gamma_i f = \sum_{j=1}^3(\delta_{ij} - n^\Gamma_i n^\Gamma_j ) \partial_j f$, $H_\Gamma = H_\Gamma (x,t) = - {\rm{div}}_\Gamma n_\Gamma$. Remark that $H_\Gamma$ is the \emph{mean curvature in the direction} $n_\Gamma$. Remark also that the conditions \eqref{eq14} mean barotropic ones, and that \eqref{eq14} correspond to the pressures derived from a thermodynamic approach (see \cite{K23}).

The second system is the following inviscid multiphase flow system with incompressible surface flow:
\begin{equation}\label{eq15}
\begin{cases}
D_t^A \rho_A + ({\rm{div}} v_A) \rho_A = 0 & \text{ in } \Omega_{A,T},\\
\rho_A D_t^A v_A + {\rm{grad}} \mathfrak{p}_A = 0 & \text{ in } \Omega_{A,T},\\
D_t^B \rho_B + ({\rm{div}} v_B) \rho_B = 0 & \text{ in } \Omega_{B,T},\\
\rho_B D_t^B v_B + {\rm{grad}} \mathfrak{p}_B = 0 & \text{ in } \Omega_{B,T},\\
D_t^S \rho_S = 0 & \text{ on } \Gamma_{T},\\
{\rm{div}}_\Gamma v_S = 0 & \text{ on } \Gamma_{T},\\
\rho_S D_t^S v_S + {\rm{grad}}_\Gamma \Pi_S + \Pi_S H_\Gamma n_\Gamma + \mathfrak{p}_B n_\Gamma - \mathfrak{p}_A n_\Gamma = 0 & \text{ on } \Gamma_{T},
\end{cases}
\end{equation}
with \eqref{eq12} and \eqref{eq14}. Remark that the condition ${\rm{div}}_\Gamma v_S =0$ means surface divergence-free one.

The third system is the following inviscid multiphase flow system with a tangential compressible surface flow:
\begin{equation}\label{eq16}
\begin{cases}
D_t^A \rho_A + ({\rm{div}} v_A) \rho_A = 0 & \text{ in } \Omega_{A,T},\\
\rho_A D_t^A v_A + {\rm{grad}} \mathfrak{p}_A = 0 & \text{ in } \Omega_{A,T},\\
D_t^B \rho_B + ({\rm{div}} v_B) \rho_B = 0 & \text{ in } \Omega_{B,T},\\
\rho_B D_t^B v_B + {\rm{grad}} \mathfrak{p}_B = 0 & \text{ in } \Omega_{B,T},\\
D_t^S \rho_S + ({\rm{div}}_\Gamma v_S) \rho_S = 0 & \text{ on } \Gamma_{T},\\
P_\Gamma \rho_S D_t^S v_S + {\rm{grad}}_\Gamma \mathfrak{p}_S= 0 & \text{ on } \Gamma_{T},\\
\mathfrak{p}_S H_\Gamma n_\Gamma + \mathfrak{p}_B n_\Gamma - \mathfrak{p}_A n_\Gamma = 0 & \text{ on } \Gamma_{T},
\end{cases}
\end{equation}
with \eqref{eq12} and \eqref{eq14}.

\begin{remark}
In order to analyze or numerical simulate systems \eqref{eq13} and \eqref{eq15}, we need some conditions on the surface flow $u_S$ or the motion velocity $w_S$ since we focus on the total velocity and pressure to make our models. System \eqref{eq17} is suitable for mathematical analysis since the motion velocity $w_S$ is given by
\begin{equation}\label{eq17}
w_S = \frac{1}{\mathfrak{p}_S H_\Gamma}\{ \mathfrak{p}_A (v_A \cdot n_\Gamma) - \mathfrak{p}_B (v_B \cdot n_\Gamma) \} n_\Gamma \text{ on } \Gamma_T
\end{equation}
if $\mathfrak{p}_S H_\Gamma \neq 0$. In fact, by \eqref{eq12} and \eqref{eq16}, we find that
\begin{align*}
\mathfrak{p}_S H_\Gamma (n_\Gamma \cdot v_S) & = \mathfrak{p}_A (n_\Gamma \cdot v_S) - \mathfrak{p}_B (n_\Gamma \cdot v_S)\\
& = \mathfrak{p}_A (v_A \cdot n_\Gamma ) - \mathfrak{p}_B (v_B \cdot n_\Gamma ) \text{ on }\Gamma_T.
\end{align*}
From $v_S \cdot n_\Gamma = w_S \cdot n_\Gamma$ and $P_\Gamma w_S = { }^t (0,0,0)$, we have \eqref{eq17}.
\end{remark}

The second purpose is to study the conservation and energy laws of systems \eqref{eq13} and \eqref{eq15}. In fact, any solution to systems \eqref{eq13} and \eqref{eq15} with \eqref{eq12} and \eqref{eq14} satisfies that for $t_1 < t_2$,
\begin{multline}\label{eq18}
\int_{\Omega_A (t_2)} \rho_A (x,t_2) { \ } d x + \int_{\Omega_B (t_2)} \rho_B (x,t_2) { \ } d x + \int_{\Gamma (t_2)} \rho_S (x,t_2) { \ }d \mathcal{H}_x^2 \\
= \int_{\Omega_A (t_1)} \rho_A (x,t_1) { \ } d x + \int_{\Omega_B (t_1)} \rho_B (x,t_1) { \ } d x + \int_{\Gamma (t_1)} \rho_S (x,t_1) { \ }d \mathcal{H}_x^2,
\end{multline}
\begin{multline}\label{eq19}
\int_{\Omega_A (t_2)} \rho_A v_A { \ }d x + \int_{\Omega_B (t_2)} \rho_B v_B { \ }d x + \int_{\Gamma (t_2)} \rho_S v_S { \ }d \mathcal{H}^2_x\\
= \int_{\Omega_A (t_1)} \rho_A v_A { \ }d x + \int_{\Omega_B (t_1)} \rho_B v_B { \ }d x + \int_{\Gamma (t_1)} \rho_S v_S { \ }d \mathcal{H}^2_x\\ - \int_{t_1}^{t_2} \int_{\partial \Omega} \mathfrak{p}_B n_\Omega { \ } d \mathcal{H}^2_x d t,
\end{multline}
and
\begin{multline}\label{eq1010}
\int_{\Omega_A (t_2)} \frac{1}{2} \rho_A  \vert v_A  \vert^2 { \ }d x + \int_{\Omega_B (t_2)} \frac{1}{2} \rho_B  \vert v_B  \vert^2 { \ }d x +\int_{\Gamma (t_2)} \frac{1}{2} \rho_S  \vert v_S  \vert^2 { \ }d \mathcal{H}_x^2\\
= \int_{\Omega_A (t_1)} \frac{1}{2} \rho_A  \vert v_A  \vert^2 { \ }d x + \int_{\Omega_B (t_1)} \frac{1}{2} \rho_B  \vert v_B  \vert^2 { \ }d x +\int_{\Gamma (t_1)} \frac{1}{2} \rho_S  \vert v_S  \vert^2 { \ }d \mathcal{H}_x^2\\
+ \int_{t_1}^{t_2 } \int_{\Omega_A (t)} ({\rm{div}} v_A) \mathfrak{p}_A { \ }d x d t + \int_{t_1}^{t_2} \int_{\Omega_B (t)} ({\rm{div}} v_B) \mathfrak{p}_B  { \ }d x dt \\
+ \int_{t_1}^{t_2} \int_{\Gamma (t)} ({\rm{div}}_\Gamma v_S) \mathfrak{p}_S { \ }d \mathcal{H}_x^2 d t.
\end{multline}
Here $d \mathcal{H}^2_x$ denotes the 2-dimensional Hausdorff measure. Moreover, any solution to system \eqref{eq13} with \eqref{eq12} and \eqref{eq14} satisfies that for $t_1 < t_2$,
\begin{multline}\label{eq1011}
\int_{\Omega_A (t_2)} \left( \frac{1}{2} \rho_A  \vert v_A  \vert^2 + p_A (\rho_A) \right) { \ }d x + \int_{\Omega_B (t_2)} \left( \frac{1}{2} \rho_B  \vert v_B  \vert^2 + p_B (\rho_B) \right) { \ }d x\\
 +\int_{\Gamma (t_2)} \left( \frac{1}{2} \rho_S  \vert v_S  \vert^2 + p_S (\rho_S) \right) { \ }d \mathcal{H}_x^2\\
= \int_{\Omega_A (t_1)} \left( \frac{1}{2} \rho_A  \vert v_A  \vert^2 + p_A (\rho_A) \right) { \ }d x + \int_{\Omega_B (t_1)} \left( \frac{1}{2} \rho_B  \vert v_B  \vert^2 + p_B (\rho_B) \right) { \ }d x\\
 +\int_{\Gamma (t_1)} \left( \frac{1}{2} \rho_S  \vert v_S  \vert^2 + p_S (\rho_S) \right) { \ }d \mathcal{H}_x^2.
\end{multline}
We often call \eqref{eq18}, \eqref{eq19}, \eqref{eq1010}, and \eqref{eq1011}, the \emph{law of conservation of mass}, the \emph{law of conservation of momentum}, the \emph{law of conservation of energy}, and the \emph{law of conservation of total energy}, respectively. See Theorem \ref{thm29} for details.

Let us state three difficulties in the derivation of our multiphase flow systems, and the key ideas to overcome these difficulties. The first difficultly is to drive the pressure of the compressible fluid on the surface $\Gamma_T$. In order to derive the pressure terms of our compressible fluid systems, we make use of the feature of the barotropic fluids. More precisely, we assume that the pressure of the compressible fluid depends only on the density of the fluid (see \eqref{eq14}). The second difficulty is to drive the pressure of the incompressible fluid on the surface. In order to derive the surface pressure of system \eqref{eq15}, we apply a generalized Helmholtz-Weyl decomposition on a closed surface (see Lemma \ref{lem72}). The third difficulty is to derive the relationship among the pressures of the fluids in the moving domains $\Omega_A(t),$ $\Omega_B(t)$, and surface $\Gamma (t)$. To overcome the difficult point, we apply an energetic variational approach to derive the relationship. An energetic variational approach is a mathematical modeling method, which had been studied by Strutt \cite{Str73} and Onsager \cite{Ons31a, Ons31b}. To derive system \eqref{eq13}, we study the variation of the following \emph{action integral}:
\begin{align*}
\int_0^T \int_{\Omega^{\varepsilon}_A (t)} \left\{ \frac{1}{2} \rho_A^{\varepsilon}  \vert  v^{\varepsilon}_A  \vert^2 - p_A (\rho^{\varepsilon}_A ) \right\} d x d t + \int_0^T \int_{\Omega^{\varepsilon}_B (t)} \left\{ \frac{1}{2} \rho_B^{\varepsilon}  \vert  v^{\varepsilon}_B  \vert^2 - p_B (\rho^{\varepsilon}_B ) \right\} d x d t\\
 + \int_0^T \int_{\Gamma^{\varepsilon} (t)} \left\{ \frac{1}{2} \rho_S^{\varepsilon}  \vert  v^{\varepsilon}_S  \vert^2 - p_S (\rho^{\varepsilon}_S ) \right\} d \mathcal{H}^2_x d t.
\end{align*}
Here $\Delta^\varepsilon$ is a variation of $\Delta$. See Theorem \ref{thm28} for details.

Let us state some results on mathematical modeling of inviscid fluid flow systems on surfaces. Arnol'd \cite{Arn66,Arn97} applied the Lie group of diffeomorphisms to derive an inviscid incompressible fluid system on a manifold. See also Ebin-Marsden \cite{EM70}. Koba-Liu-Giga \cite{KLG17} employed an energetic variational approach to derive an inviscid incompressible fluid system on an evolving closed surface. Koba \cite{K18} made use of an energetic variational approach to derive an inviscid compressible fluid system on an evolving closed surface. This paper improves the methods in \cite{KLG17, K18} to derive the inviscid multiphase flow systems with surface flow and tension.

Finally, we introduce the results related to this paper. Serrin \cite{Ser59} introduced Euler's ideas and Cauchy's principle to derive the inviscid fluid system in a domain. Gyarmati \cite{Gya70}, Hyon-Kwak-Liu \cite{HKL10}, Koba-Sato \cite{KS17} employed their energetic variational approaches to make and study several models for fluid dynamics in domains. Bothe-Pr\"{u}ss \cite{BP12} applied the \emph{Boussinesq-Scriven law} to make their model for multiphase flow with surface tension and viscosities. See \cite{BP12,K18} for the Boussinesq-Scriven law. Koba \cite{K23} applied the first law of thermodynamics to derive their multiphase flow system with surface flow and tension. Our approach is different from one in \cite{BP12, K23}. See also Slattery-Sagis-Oh \cite{SSO07}, Gatignol-Prud'homme \cite{GP01} for interfacial phenomena.

The outline of this paper is as follows: In Section \ref{sect2} we first introduce flow maps in $\overline{\Omega_T}$, and we state the main results of this paper. In Section \ref{sect3} we study some properties of the Riemannian metrics determined by flow maps in $\overline{\Omega_T}$. In Section \ref{sect4} we calculate variations of flow maps to our action integrals determined by the kinetic energies $(\rho_A  \vert v_A \vert^2/2, \rho_B  \vert v_B  \vert^2/2,\rho_S  \vert v_S  \vert^2/2)$. In Section \ref{sect5}, we apply an energetic variational approach to make mathematical models for inviscid multiphase flow. In Section \ref{sect6} we investigate the conservation and energy laws of our systems. In Appendix, we give some useful tools to analyze functions on surfaces.

\section{Variations of flow maps and main results}\label{sect2}

We first introduce variations $(\tilde{x}_A^\varepsilon, \tilde{x}_B^\varepsilon, \tilde{x}^\varepsilon_S)$ of flow maps $(\tilde{x}_A, \tilde{x}_B, \tilde{x}_S)$ in $\overline{\Omega_T}$. The flow maps $(\tilde{x}_A, \tilde{x}_B, \tilde{x}_S)$ and its variations $(\tilde{x}_A^\varepsilon, \tilde{x}_B^\varepsilon, \tilde{x}^\varepsilon_S)$ are essential tools to make mathematical models for inviscid multiphase flow. Then we state the main results of this paper.

Let us explain the conventions used in this paper. We use the italic characters $i,j, k, \ell$ as $1,2,3$, and the Greek characters $\alpha , \beta$ as $1,2$, that is, $i, j , k , \ell \in \{ 1 , 2, 3\} $ and $\alpha , \beta \in\{ 1 ,2 \}$.

We first define our moving domains and surface
\begin{definition}[Moving domains and surface]\label{def21}
Let $\Omega $ be a bounded domain in $\mathbb{R}^3$ with a $C^\infty$-boundary $\partial \Omega$. Let $\Omega_A (t)(= \{ \Omega_A (t) \}_{0 \leq t <T})$ be a $C^\infty$-bounded domain in $\mathbb{R}^3$ depending on time $t \in [0,T)$ such that $\Omega_A (t) \Subset \Omega$. Set
\begin{equation*}
\Omega_B (t) = \Omega \setminus \overline{\Omega_A (t)} \text{ and } \Gamma (t) = \partial \Omega_A (t).
\end{equation*}
For each $0 \leq t < T$, assume that $\Gamma (t)$ is a closed Riemannian 2-dimensional manifold. Define $\Omega_{A,T}$, $\Omega_{B,T}$, $\Gamma_T$, $\Omega_T$, and $\partial \Omega_T$ by \eqref{eq11}. Set
\begin{multline*}
\overline{\Omega_{A,T}} = \bigcup_{0 \leq t < T} \{ \overline{\Omega_A (t)} \times \{ t \} \},{ \ } \overline{\Omega_{B,T}} = \bigcup_{0 \leq t < T} \{ \overline{\Omega_B (t)} \times \{ t \} \},\\
\overline{\Gamma_T} = \bigcup_{0 \leq t < T} \{ \Gamma (t) \times \{ t \} \}, { \ }\overline{\Omega_T} = \overline{\Omega} \times [0,T).
\end{multline*}
\end{definition}
Let $\Omega$, $\{ \Omega_A (t) \}_{0 \leq t <T}$ be bounded domains satisfying the properties as in Definition \ref{def21}. From now we fix $\Omega$ and $\Omega_A (t)$. Note that $\Omega_B (t)$ is a bounded domain in $\mathbb{R}^3$ with $C^\infty$-boundaries $\partial \Omega$ and $\Gamma (t)$ (see Figure \ref{Fig1}).

Next we introduce function spaces in moving domains and surfaces.
\begin{definition}[Function spaces in moving domains and surfaces]\label{def22}
Let $\mathcal{S} \subset \mathbb{R}^4$. Define
\begin{equation*}
C^\infty (\mathcal{S}) = \{ f: \mathcal{S} \to \mathbb{R};  { \ } f = F \vert_{\mathcal{S}} \text{ for some }F \in C^\infty (\mathbb{R}^4) \}.
\end{equation*}
For example, 
\begin{align*}
C^\infty (\Omega_{A,T}) & = \{ f: \Omega_{A,T} \to \mathbb{R};  { \ } f = F \vert_{\Omega_{A,T}}  \text{ for some }F \in C^\infty (\mathbb{R}^4) \},\\
C^\infty (\overline{\Omega_{B,T}}) & = \{ f: \overline{\Omega_{B,T}} \to \mathbb{R};  { \ } f = F \vert_{\overline{\Omega_{B,T}}}  \text{ for some }F \in C^\infty (\mathbb{R}^4) \},\\
C^\infty (\Gamma_{T})  & = \{ f: \Gamma_{T} \to \mathbb{R};  { \ } f = F \vert_{\Gamma_T}  \text{ for some }F \in C^\infty (\mathbb{R}^4) \}.
\end{align*}
\end{definition}
\noindent Remark that function spaces in fixed domains and surfaces in $\mathbb{R}^3$ are usual function spaces.

Now we introduce flow maps $(\tilde{x}_A, \tilde{x}_B, \tilde{x}_S)$ in $\overline{\Omega_T}$.
\begin{definition}[Flow maps in $\overline{\Omega_T}$]\label{def23}{ \ }\\
$(\rm{i})$ Let $\tilde{x}_A = \tilde{x}_A(\xi_A , t) = { }^t ( \tilde{x}_1^A(\xi_A, t) , \tilde{x}_2^A (\xi_A , t) , \tilde{x}_3^A ( \xi_A , t) )$ be in $[C^\infty(\overline{\Omega_A (0)} \times [0,T) )]^3$. We call $\tilde{x}_A$ a \emph{flow map} in $\overline{\Omega_{A,T}}$ if the following three properties hold:\\
{\bf{Property}} $(\rm{I})$ For $0<t <T$,
\begin{align*}
\Omega_A (t) = \{ x = { }^t (x_1,x_2,x_3) \in \mathbb{R}^3; { \ } x = \tilde{x}_A(\xi_A , t), { \ }\xi_A \in \Omega_A (0)  \},\\
\overline{\Omega_A (t)} = \{ x = { }^t (x_1,x_2,x_3) \in \mathbb{R}^3; { \ } x = \tilde{x}_A(\xi_A , t),{ \ } \xi_A \in \overline{\Omega_A (0)}  \}.
\end{align*}
{\bf{Property}} $(\rm{II})$ There exists a smooth function $v_A = v_A (x,t)=$\\ $ { }^t (v^A_1 (x,t) , v^A_2 (x,t),v^A_3 (x,t))$ such that for every $\xi_A \in \overline{\Omega_A (0)}$,
\begin{equation*}
\begin{cases}
\partial_t \tilde{x}_A (\xi_A , t) = v_A ( \tilde{x}_A (\xi_A , t ) , t) , { \ }t \in (0,T),\\
\tilde{x}_A (\xi_A , 0) = \xi_A.
\end{cases}
\end{equation*}
We call $v_A$ the \emph{velocity determined by the flow map} $\tilde{x}_A$.\\
{\bf{Property}} $(\rm{III})$ For each $0<t <T$ and $\Lambda_A (t) \subset \overline{\Omega_A (t)}$, there is $\mathfrak{M}_A \subset \overline{\Omega_A (0)}$ such that
\begin{align*}
\Lambda_A (t) = \{ x = { }^t (x_1,x_2,x_3) \in \mathbb{R}^3; { \ } x = \tilde{x}_A(\xi_A , t), { \ }\xi_A \in \mathfrak{M}_A  \}.
\end{align*}
$(\rm{ii})$ As in $(\rm{i})$, we define a \emph{flow map} $\tilde{x}_B$ in $\overline{\Omega_{B,T}}$ and the \emph{velocity} $v_B$ \emph{determined by the flow map} $\tilde{x}_B$.\\
$(\rm{iii})$ Let $\tilde{x}_S = \tilde{x}_S (\xi_S , t) = { }^t ( \tilde{x}_1^S(\xi_S, t) , \tilde{x}_2^S (\xi_S , t) , \tilde{x}_3^S ( \xi_S , t) )$ be in $[C^\infty(\Gamma (0) \times [0,T) )]^3$. We call $\tilde{x}_S$ a \emph{flow map} in $\overline{\Gamma_T}$ if the following three properties hold:\\
{\bf{Property}} $(\rm{I})$ For $0<t <T$,
\begin{equation*}
\Gamma (t) = \{ x = { }^t (x_1,x_2,x_3) \in \mathbb{R}^3; { \ } x = \tilde{x}_S (\xi_S , t), { \ }\xi_S \in \Gamma (0)  \}.
\end{equation*}
{\bf{Property}} $(\rm{II})$ There exists a smooth function $v_S = v_S (x,t)=$\\ ${ }^t (v^S_1 (x,t) , v^S_2 (x,t),v^S_3 (x,t))$ such that for every $\xi_S \in \Gamma (0)$,
\begin{equation*}
\begin{cases}
\partial_t \tilde{x}_S (\xi_S , t) = v_S ( \tilde{x}_S (\xi_S , t ) , t) , { \ }t \in (0,T),\\
\tilde{x}_S (\xi_S , 0) = \xi_S.
\end{cases}
\end{equation*}
We call $v_S$ the \emph{velocity determined by the flow map} $\tilde{x}_S$.\\
{\bf{Property}} $(\rm{III})$ For each $0<t <T$ the mapping $\tilde{x}_S(\cdot , t) : \Gamma (0) \to \Gamma (t)$ is bijective.
\end{definition}

Let $(\tilde{x}_A , \tilde{x}_B , \tilde{x}_S)$ be flow maps in $\overline{\Omega_{T}}$, and $(v_A , v_B , v_S) $ be the velocities determined by the flow maps $(\tilde{x}_A , \tilde{x}_B , \tilde{x}_S)$ satisfying the properties as in Definition \ref{def23}. From now we fix $(\tilde{x}_A , \tilde{x}_B , \tilde{x}_S)$ and $(v_A , v_B , v_S)$.

Next we introduce two types of variations. The first one is a variation of our domains and surface. The second one is a variation of flow maps in a variation of our domains and surface.
\begin{definition}[Variations of domains and surface]\label{def24}
For each $- 1 < \varepsilon <1$, let $\Omega^{\varepsilon}_A (t)(= \{ \Omega^{\varepsilon}_A (t) \}_{0 \leq t <T})$ be a $C^\infty$-bounded domain in $\mathbb{R}^3$ depending on time $t \in [0,T)$ such that $\Omega^{\varepsilon}_A (t) \Subset \Omega$. For each $0 \leq t < T$, assume that $\Gamma^\varepsilon (t)$ is a closed Riemannian 2-dimensional manifold. Set
\begin{equation*}
\Omega^{\varepsilon}_B (t) = \Omega \setminus \overline{\Omega^{\varepsilon}_A (t)} \text{ and } \Gamma^{\varepsilon} (t) = \partial \Omega^{\varepsilon}_A (t).
\end{equation*}
Write
\begin{multline*}
\Omega^\varepsilon_{A,T} = \bigcup_{0 < t < T} \{ \Omega^\varepsilon_A (t) \times \{ t \} \},{ \ } \Omega^\varepsilon_{B,T} = \bigcup_{0 < t < T} \{ \Omega^\varepsilon_B (t) \times \{ t \} \},\\
\Gamma^\varepsilon_T = \bigcup_{0 < t < T} \{ \Gamma^\varepsilon (t) \times \{ t \} \},{ \ }\overline{\Omega^\varepsilon_{A,T}} = \bigcup_{0 \leq t < T} \{ \overline{\Omega^\varepsilon_A (t)} \times \{ t \} \},\\
{ \ } \overline{\Omega^\varepsilon_{B,T}} = \bigcup_{0 \leq t < T} \{ \overline{\Omega^\varepsilon_B (t)} \times \{ t \} \},{ \ }\overline{\Gamma^\varepsilon_T} = \bigcup_{0 \leq t < T} \{ \Gamma^\varepsilon (t) \times \{ t \} \}.
\end{multline*}
We say that $(\Omega_{A,T}^{\varepsilon}, \Omega^{\varepsilon}_{B,T} , \Gamma^{\varepsilon}_T)$ is a variation of $(\Omega_{A,T}, \Omega_{B,T} , \Gamma_T )$ if $\Omega_A^{\varepsilon}(0) = \Omega_A (0)$, $\Omega_B^{\varepsilon}(0) = \Omega_B (0)$, $\Gamma^{\varepsilon}(0) = \Gamma (0)$, $\Omega^{\varepsilon}_A (t) \vert_{\varepsilon =0} = \Omega_A (t)$, $\Omega^{\varepsilon}_B (t) \vert_{\varepsilon =0} = \Omega_B (t)$, and $\Gamma^{\varepsilon} (t) \vert_{\varepsilon =0} = \Gamma (t)$.
\end{definition}
\noindent Note that $\Omega = \Omega^\varepsilon_A(t) \bigcup \Gamma^\varepsilon (t) \bigcup \Omega_B^\varepsilon (t)$ (see Figure \ref{Fig1}). In this paper, we assume that the dependence of $\Omega_A^\varepsilon (t)$, $\Omega_B^\varepsilon (t)$, and $\Gamma^\varepsilon (t)$ is smooth with respect to the parameter $\varepsilon$. From now the symbol $n_{\Gamma^\varepsilon} = n_{\Gamma^\varepsilon} ( x , t ) = { }^t (n^{\Gamma^\varepsilon}_1 , n^{\Gamma^\varepsilon}_2 , n^{\Gamma^\varepsilon}_3)$ denotes the unit outer normal vector at $x \in \Gamma^\varepsilon (t)$.

\begin{definition}[Flow maps in $(\overline{\Omega^\varepsilon_{A,T}} , \overline{\Omega^\varepsilon_{B,T}} , \overline{\Gamma^\varepsilon_T})$]\label{def25} Let $- 1 < \varepsilon <1$, and let $(\Omega^\varepsilon_{A,T} , \Omega^\varepsilon_{B,T} , \Gamma^\varepsilon_T)$ be a variation of $(\Omega_{A,T} , \Omega_{B,T} , \Gamma_T)$.\\
\noindent $(\rm{i})$ Let $\tilde{x}^\varepsilon_A = \tilde{x}^\varepsilon_A(\xi_A , t) = { }^t ( \tilde{x}_1^{A,\varepsilon}(\xi_A, t) , \tilde{x}_2^{A,\varepsilon} (\xi_A , t) , \tilde{x}_3^{A,\varepsilon} ( \xi_A , t) )$ be in $[C^\infty(\overline{\Omega_A (0)} \times [0,T) )]^3$. We call $\tilde{x}^{\varepsilon}_A$ a \emph{flow map} in $\overline{\Omega^\varepsilon_{A,T}}$ if the following three properties hold:\\
{\bf{Property}} $(\rm{I})$ For $0<t <T$,
\begin{align*}
\Omega^\varepsilon_A (t) = \{ x = { }^t (x_1,x_2,x_3) \in \mathbb{R}^3; { \ } x = \tilde{x}^\varepsilon_A(\xi_A , t), { \ }\xi_A \in \Omega_A (0)  \},\\
\overline{\Omega^\varepsilon_A (t)} = \{ x = { }^t (x_1,x_2,x_3) \in \mathbb{R}^3; { \ } x = \tilde{x}^\varepsilon_A(\xi_A , t),{ \ } \xi_A \in \overline{\Omega_A (0)}  \}.
\end{align*}
{\bf{Property}} $(\rm{II})$ There exists a smooth function $v^\varepsilon_A = v^\varepsilon_A (x,t)=$\\ $ { }^t (v^{A,\varepsilon}_1 (x,t) , v^{A,\varepsilon}_2 (x,t),v^{A,\varepsilon}_3 (x,t))$ such that for every $\xi_A \in \overline{\Omega_A (0)}$,
\begin{equation*}
\begin{cases}
\partial_t \tilde{x}^\varepsilon_A (\xi_A , t) = v^\varepsilon_A ( \tilde{x}^\varepsilon_A (\xi_A , t ) , t) , { \ }t \in (0,T),\\
\tilde{x}^\varepsilon_A (\xi_A , 0) = \xi_A.
\end{cases}
\end{equation*}
We call $v^\varepsilon_A$ the \emph{velocity determined by the flow map} $\tilde{x}^\varepsilon_A$.\\
{\bf{Property}} $(\rm{III})$ For each $0<t <T$ and $\Lambda^\varepsilon_A (t) \subset \overline{\Omega^\varepsilon_A (t)}$, there is $\mathfrak{M}_A \subset \overline{\Omega_A (0)}$ such that
\begin{align*}
\Lambda^\varepsilon_A (t) = \{ x = { }^t (x_1,x_2,x_3) \in \mathbb{R}^3; { \ } x = \tilde{x}^\varepsilon_A(\xi_A , t), { \ }\xi_A \in \mathfrak{M}_A  \}.
\end{align*}
$(\rm{ii})$ As in $(\rm{i})$, we define a \emph{flow map} $\tilde{x}^\varepsilon_B$ in $\overline{\Omega^\varepsilon_{B,T}}$ and the \emph{velocity} $v^\varepsilon_B$ \emph{determined by the flow map} $\tilde{x}^\varepsilon_B$.\\
$(\rm{iii})$ Let $\tilde{x}^\varepsilon_S = \tilde{x}^\varepsilon_S(\xi_S , t) = { }^t ( \tilde{x}_1^{S , \varepsilon}(\xi_S, t) , \tilde{x}_2^{S,\varepsilon} (\xi_S , t) , \tilde{x}_3^{S,\varepsilon} ( \xi_S , t) )$ be in $[C^\infty (\Gamma (0) \times [0,T) )]^3$. We call $\tilde{x}^\varepsilon_S$ a \emph{flow map} in $\overline{\Gamma^\varepsilon_T}$ if the following three properties hold:\\
{\bf{Property}} $(\rm{I})$ For $0<t <T$,
\begin{equation*}
\Gamma^\varepsilon (t) = \{ x = { }^t (x_1,x_2,x_3) \in \mathbb{R}^3; { \ } x = \tilde{x}^\varepsilon_S (\xi_S , t), { \ }\xi_S \in \Gamma (0)  \}.
\end{equation*}
{\bf{Property}} $(\rm{II})$ There exists a smooth function $v^\varepsilon_S = v^\varepsilon_S (x,t)=$\\ $ { }^t (v^{S,\varepsilon}_1 (x,t) , v^{S, \varepsilon}_2 (x,t),v^{S,\varepsilon}_3 (x,t))$ such that for every $\xi_S \in \Gamma (0)$,
\begin{equation*}
\begin{cases}
\partial_t \tilde{x}^\varepsilon_S (\xi_S , t) = v^\varepsilon_S ( \tilde{x}^\varepsilon_S (\xi_S , t ) , t) , { \ }t \in (0,T),\\
\tilde{x}^\varepsilon_S (\xi_S , 0) = \xi_S.
\end{cases}
\end{equation*}
We call $v^\varepsilon_S$ the \emph{velocity determined by the flow map} $\tilde{x}^\varepsilon_S$.\\
{\bf{Property}} $(\rm{III})$ For each $0<t <T$ the mapping $\tilde{x}^\varepsilon_S(\cdot , t) : \Gamma (0) \to \Gamma^\varepsilon (t)$ is bijective.
\end{definition}

For each $-1 < \varepsilon < 1$, from now, the symbol $(\Omega^\varepsilon_{A,T} , \Omega^\varepsilon_{B,T} , \Gamma^\varepsilon_T)$ is a variation of $(\Omega_{A,T} , \Omega_{B,T} , \Gamma_T)$, $(\tilde{x}^\varepsilon_A , \tilde{x}_B^\varepsilon , \tilde{x}_S^\varepsilon)$ denote flow maps in $(\overline{\Omega^\varepsilon_{A,T}} , \overline{\Omega^\varepsilon_{B,T}} , \overline{\Gamma^\varepsilon_T})$, and $(v_A^\varepsilon , v_B^\varepsilon, v_S^\varepsilon)$ denote the velocities determined by the flow map $(\tilde{x}^\varepsilon_A , \tilde{x}^\varepsilon_B , \tilde{x}^\varepsilon_S)$ satisfying the properties as in Definitions \ref{def24} and \ref{def25}. 

\begin{definition}[Variations of flow maps in $(\overline{\Omega^\varepsilon_{A,T}} , \overline{\Omega^\varepsilon_{B,T}} , \overline{\Gamma^\varepsilon_T})$]\label{def26} { \ }\\
Let $z_A = z_A (x,t) = { }^t (z^A_1,z_2^A, z_3^A)$, $z_B = z_B (x,t) = { }^t (z^B_1,z_2^B, z_3^B)$, $z_S = z_S (x,t) = { }^t (z^S_1,z_2^S, z_3^S)$ be smooth functions in $\mathbb{R}^4$. We say that $(z_A, z_B, z_S)$ is a variation of $(\tilde{x}^{\varepsilon}_A, \tilde{x}^{\varepsilon}_B , \tilde{x}^{\varepsilon}_S)$ if there are smooth functions $\tilde{y}_A = \tilde{y}_A(\xi_A ,t ) = { }^t (\tilde{y}_1^A, \tilde{y}_2^A, \tilde{y}_3^A)$, $\tilde{y}_B = \tilde{y}_B (\xi_B ,t ) = { }^t (\tilde{y}_1^B, \tilde{y}_2^B, \tilde{y}_3^B)$, $\tilde{y}_S = \tilde{y}_S(\xi_S ,t ) = { }^t (\tilde{y}_1^S, \tilde{y}_2^S, \tilde{y}_3^S)$ such that for $0 \leq t <T$, $\xi_A \in \overline{\Omega_A(0)}$, $\xi_B \in \overline{\Omega_B(0)}$, $\xi_S \in \Gamma (0)$,
\begin{align*}
&\frac{d}{d \varepsilon} \bigg \vert_{\varepsilon =0} \tilde{x}_A^\varepsilon (\xi_A ,t ) = \tilde{y}_A (\xi_A , t ),{ \ }z_A ( \tilde{x}_A (\xi_A, t ) , t) = \tilde{y}_A (\xi_A , t),\\
&\frac{d}{d \varepsilon} \bigg \vert_{\varepsilon =0} \tilde{x}_B^\varepsilon (\xi_B ,t ) = \tilde{y}_B (\xi_B , t ),{ \ }z_B ( \tilde{x}_B (\xi_B, t ) , t) = \tilde{y}_B (\xi_B , t),\\
&\frac{d}{d \varepsilon} \bigg \vert_{\varepsilon =0} \tilde{x}_S^\varepsilon (\xi_S ,t ) = \tilde{y}_S (\xi_S , t ),{ \ }z_S ( \tilde{x}_S (\xi_S, t ) , t) = \tilde{y}_S (\xi_S , t).
\end{align*}
\end{definition}

Before stating the main results of this paper, we recall the transport theorems. For each $- 1 < \varepsilon < 1$, let $\rho_A$, $\rho_B$, $\rho_S$, $\rho^\varepsilon_A$, $\rho^\varepsilon_B$, $\rho^\varepsilon_S$ be smooth functions in $\mathbb{R}^4$. 
\begin{proposition}[Continuity equations]\label{prop27}The following two assertions hold:\\ 
\noindent $(\rm{i})$ Assume that for every $0< t <T$ and $\Lambda \subset \Omega$,
\begin{equation*}
\frac{d}{dt} \bigg( \int_{\Omega_A (t) \cap \Lambda} \rho_A { \ } d x + \int_{\Omega_B(t) \cap \Lambda} \rho_B { \ }d x + \int_{\Gamma (t) \cap \Lambda} \rho_S { \ }d \mathcal{H}_x^2 \bigg) = 0 . 
\end{equation*}
Then $(\rho_A , \rho_B , \rho_S )$ satisfies
\begin{equation}\label{eq21}
\begin{cases}
D_t^A \rho_A + ({\rm{div}} v_A) \rho_A = 0 & \text{ in }\Omega_{A,T},\\
D_t^B \rho_B + ({\rm{div}} v_B) \rho_B = 0 & \text{ in }\Omega_{B,T},\\
D_t^S \rho_S + ({\rm{div}}_\Gamma v_S) \rho_S = 0 & \text{ on }\Gamma_{T}.
\end{cases}
\end{equation}
\noindent $(\rm{ii})$ Let $-1< \varepsilon < 1$. Assume that for every $0< t <T$ and $\Lambda \subset \Omega$,
\begin{equation*}
\frac{d}{dt} \bigg( \int_{\Omega^\varepsilon_A (t) \cap \Lambda} \rho^\varepsilon_A { \ } d x + \int_{\Omega^\varepsilon_B(t) \cap \Lambda} \rho^\varepsilon_B { \ }d x + \int_{\Gamma^\varepsilon (t) \cap \Lambda} \rho^\varepsilon_S { \ }d \mathcal{H}_x^2 \bigg) = 0 . 
\end{equation*}
Then $(\rho^\varepsilon_A , \rho^\varepsilon_B , \rho^\varepsilon_S )$ satisfies
\begin{equation}\label{eq22}
\begin{cases}
D_t^{A,\varepsilon} \rho^\varepsilon_A + ({\rm{div}} v^\varepsilon_A) \rho^\varepsilon_A = 0 & \text{ in }\Omega^\varepsilon_{A,T},\\
D_t^{B,\varepsilon} \rho^\varepsilon_B + ({\rm{div}} v^\varepsilon_B) \rho^\varepsilon_B = 0 & \text{ in }\Omega^\varepsilon_{B,T},\\
D_t^{S,\varepsilon} \rho^\varepsilon_S + ({\rm{div}}_{\Gamma^\varepsilon} v^\varepsilon_S) \rho^\varepsilon_S = 0 & \text{ on }\Gamma^\varepsilon_{T}.
\end{cases}
\end{equation}
Here $D_t^{A,\varepsilon} f := \partial_t f + (v^\varepsilon_A, \nabla )f $, $D_t^{B,\varepsilon} f := \partial_t f + (v^\varepsilon_B, \nabla )f $, $D_t^{S,\varepsilon} f := \partial_t f + (v^\varepsilon_S, \nabla )f $, and ${\rm{div}}_{\Gamma^\varepsilon} v_S = \partial_1^{\Gamma^\varepsilon} v^S_1 + \partial_2^{\Gamma^\varepsilon} v^S_2 + \partial_3^{\Gamma^\varepsilon} v^S_3$, where $\partial^{\Gamma^\varepsilon}_i f := \sum_{j=1}^3(\delta_{ij} - n^{\Gamma^\varepsilon}_i n^{\Gamma^\varepsilon}_j ) \partial_j f$.
\end{proposition}
\noindent The proof of Proposition \ref{prop27} can be founded in \cite{Bet86,GSW89,DE07,KLG17, KS17}. We can prove Proposition \ref{prop27} by applying Lemma \ref{lem33}.

Now we state the main results of this paper. Let $\rho_0^A \in C^\infty ( \overline{\Omega_A (0)})$, $\rho_0^B \in C^\infty ( \overline{\Omega_B (0)})$, and $\rho_0^S \in C^\infty ( \Gamma (0))$.
Assume that $(\rho_A , \rho_B , \rho_S )$ satisfies \eqref{eq21} and
\begin{equation}\label{eq23}
\begin{cases}
\rho_A \vert_{t =0} =\rho_0^A & \text{ in } \Omega_A (0),\\
\rho_B \vert_{t =0} =\rho_0^B & \text{ in } \Omega_B (0),\\
\rho_S \vert_{t =0} =\rho_0^S & \text{ on } \Gamma (0).
\end{cases}
\end{equation}
For each $- 1 < \varepsilon < 1$ assume that $(\rho^\varepsilon_A , \rho^\varepsilon_B , \rho^\varepsilon_S )$ satisfies \eqref{eq22} and
\begin{equation}\label{eq24}
\begin{cases}
\rho^{\varepsilon}_A \vert_{t =0} =\rho_0^A & \text{ in } \Omega_A (0),\\
\rho^{\varepsilon}_B \vert_{t =0} =\rho_0^B & \text{ in } \Omega_B (0),\\
\rho^{\varepsilon}_S \vert_{t =0} =\rho_0^S & \text{ on } \Gamma (0).
\end{cases}
\end{equation}
Let $p_A, p_B , p_S \in C^1(\mathbb{R})$. For each variation $(\tilde{x}_A^\varepsilon, \tilde{x}_B^\varepsilon, \tilde{x}_S^\varepsilon)$ of the flow map $(\tilde{x}_A, \tilde{x}_B, \tilde{x}_S)$,
\begin{multline*}
A [\tilde{x}_A^\varepsilon, \tilde{x}_B^\varepsilon, \tilde{x}_S^\varepsilon] := \int_0^T \int_{\Omega^{\varepsilon}_A (t)} \left\{ \frac{1}{2} \rho^{\varepsilon}_A  \vert v^{\varepsilon}_A \vert^2 - p_A (\rho^{\varepsilon}_A) \right\}{ \ }d x dt\\
 + \int_0^T \int_{\Omega^{\varepsilon}_B (t)} \left\{ \frac{1}{2} \rho^{\varepsilon}_B  \vert v^{\varepsilon}_B \vert^2 - p_B (\rho^{\varepsilon}_B) \right\} { \ }d x dt\\
 + \int_0^T \int_{\Gamma^{\varepsilon} (t)} \left\{ \frac{1}{2} \rho^{\varepsilon}_S  \vert v^{\varepsilon}_S \vert^2 - p_S (\rho^{\varepsilon}_S) \right\} { \ }d \mathcal{H}^2_x dt.
\end{multline*}
\noindent We call $A [\tilde{x}_A^\varepsilon, \tilde{x}_B^\varepsilon, \tilde{x}_S^\varepsilon]$ and $(\rho_* \vert v_* \vert^2/2, p_* (\rho_*) )$ the \emph{action integral} and \emph{energy densities} for our models. Moreover, we assume that the Riemannian metrics $\sqrt{G_A}$, $\sqrt{G^\varepsilon_A}$, $\sqrt{G_B}$, $\sqrt{G_B^\varepsilon}$, $\sqrt{G_S}$, $\sqrt{G_S^\varepsilon}$ are positive functions (see Section \ref{sect3} for details). 
\begin{theorem}[Variations of the flow maps to the action integral]\label{thm28}
Assume that for every $0 \leq t <T$, $\xi_A \in \overline{\Omega_A(0)}$, $\xi_B \in \overline{\Omega_B (0)}$, and $\xi _S \in \Gamma (0)$,
\begin{equation}\label{eq25}
\begin{cases}
\rho_A^\varepsilon ( \tilde{x}_A^\varepsilon (\xi_A , t) , t)  \vert_{\varepsilon = 0} & = \rho_A ( \tilde{x}_A (\xi_A , t ) , t),\\
\rho_B^\varepsilon ( \tilde{x}_B^\varepsilon (\xi_B , t) , t)  \vert_{\varepsilon = 0} & = \rho_B ( \tilde{x}_B (\xi_B , t ) , t),\\
\rho_S^\varepsilon ( \tilde{x}_S^\varepsilon (\xi_S , t) , t)  \vert_{\varepsilon = 0} & = \rho_S ( \tilde{x}_S (\xi_S , t ) , t),
\end{cases}
\end{equation}
and
\begin{equation}\label{eq26}
\begin{cases}
z_A (\tilde{x}_A (\xi_A,T-),T-) ={ }^t (0,0,0),\\
z_B (\tilde{x}_B (\xi_B,T-),T-) ={ }^t (0,0,0),\\
z_S (\tilde{x}_S (\xi_S,T-),T-) ={ }^t (0,0,0).
\end{cases}
\end{equation}
Then
\begin{multline}\label{eq27}
\frac{d}{d \varepsilon} \bigg \vert_{\varepsilon = 0} A [\tilde{x}_A^\varepsilon, \tilde{x}_B^\varepsilon, \tilde{x}_S^\varepsilon] = - \int_0^T \int_{\Omega_A (t)} \{ ( \rho_A D^A_t v_A + \nabla \mathfrak{p}_A) \cdot z_A \} (x,t) { \ } d x d t\\
- \int_0^T \int_{\Omega_B (t)} (\rho_B D_t^B v_B + \nabla \mathfrak{p}_B) \cdot z_B { \ } d x d t\\
- \int_0^T \int_{\Gamma (t)} ( \rho_S D_t^S v_S + \nabla_\Gamma \mathfrak{p}_S + \mathfrak{p}_S H_\Gamma n_\Gamma ) \cdot z_S { \ } d \mathcal{H}^2_x d t\\
+ \int_0^T \int_{\Gamma (t)} (\mathfrak{p}_A n_\Gamma ) \cdot z_A { \ } d \mathcal{H}_x^2 d t - \int_0^T \int_{\Gamma (t)} (\mathfrak{p}_B n_\Gamma) \cdot z_B { \ }d \mathcal{H}^2_x dt \\ 
 + \int_0^T \int_{\partial \Omega} ( \mathfrak{p}_B n_\Omega ) \cdot z_B { \ }d \mathcal{H}^2_x d t.
\end{multline}
Here $(z_A, z_B, z_S)$ is a variation of $(\tilde{x}^{\varepsilon}_A, \tilde{x}^{\varepsilon}_B , \tilde{x}^{\varepsilon}_S)$, $(\mathfrak{p}_A, \mathfrak{p}_B, \mathfrak{p}_S)$ is define by \eqref{eq14}, and $f(T-):=\lim_{t \uparrow T} f(t)$.
\end{theorem}
\noindent Applying \eqref{eq27}, we make mathematical models for inviscid multiphase flow with surface flow. See Section \ref{sect5} for details.

\begin{theorem}[Conservation and energy laws]\label{thm29}{ \ }\\
\noindent $(\rm{i})$ Any solution to system \eqref{eq13} with \eqref{eq12} and \eqref{eq14} satisfies \eqref{eq18}, \eqref{eq19}, \eqref{eq1010}, and \eqref{eq1011}.\\
\noindent $(\rm{ii})$ Any solution to system \eqref{eq15} with \eqref{eq12} and \eqref{eq14} satisfies \eqref{eq18}, \eqref{eq19}, and \eqref{eq1010}.
\end{theorem}
We prove Theorem \ref{thm28} in Section \ref{sect4} and Theorem \ref{thm29} in Section \ref{sect6}.

\section{Preliminaries}\label{sect3}
Let us first recall some properties of the Riemannian metrics determined by flow maps. Then we study the representation formulas for our energy densities.

Let $\tilde{x}_A = \tilde{x}_A(\xi_A , t)$ be the flow map in $\overline{\Omega_{A,T}}$, and $v_A = v_A (x,t)$ be the velocity determined by the flow map $\tilde{x}_A$, that is, for every $\xi_A \in \overline{\Omega_A (0)}$ and $0< t <T$,
\begin{equation*}
\begin{cases}
\tilde{x}_{A} = \tilde{x}_{A} ( \xi_{A} , t) = { }^t ( \tilde{x}^{A}_1 (\xi_A , t) , \tilde{x}^{A}_2 (\xi_{A}, t), \tilde{x}^{A}_3 (\xi_{A},t) ),\\
v_{A} = v_{A} (x,t) = { }^t ( v^{A}_1 (x,t) , v^{A}_2 (x,t) , v^{A}_3 (x,t) ),\\
\partial_t \tilde{x}_{A} = v_A ( \tilde{x}_{A}(\xi_{A}, t) , t),\\
\tilde{x}_{A} (\xi_{A}, 0 ) = \xi_{A}.
\end{cases} 
\end{equation*}
For each $i,j \in \{ 1, 2 ,3 \}$,
\begin{equation*}
g_i^{A} = g_i^{A} ( \xi_{A} , t) := \frac{\partial \tilde{x}_{A}}{\partial \xi^A_{i}} = { }^t \left( \frac{\partial \tilde{x}^{A}_1 }{\partial \xi^{A}_i}, \frac{\partial \tilde{x}^{A}_2 }{\partial \xi^{A}_i}, \frac{\partial \tilde{x}^{A}_3 }{\partial \xi^{A}_i} \right),
\end{equation*}
\begin{equation*}
g^{A}_{ij} = g^{A}_{ij}(\xi_{A} , t) := g^{A}_i \cdot g^{A}_j = \sum_{\ell=1}^3 \frac{\partial \tilde{x}^{A}_\ell}{\partial \xi^{A}_i}\frac{\partial \tilde{x}^{A}_\ell}{\partial \xi^{A}_j},
\end{equation*}
and
\begin{equation*}
G_{A} := G_{A} (\xi_{A} , t) = {\rm{det}} (g_{ij}^{A})_{3 \times 3}.
\end{equation*}

For the flow map $\tilde{x}_B$ in $\overline{\Omega_{B,T}}$, we define $g_i^{B}$, $g^{B}_{ij}$, $G_B$ similarly to $g_i^{A}$, $g^{A}_{ij}$, $G_A$.

For each $-1 < \varepsilon < 1$, let $\tilde{x}^{\varepsilon}_A = \tilde{x}^{\varepsilon}_A( \xi_A , t)$ be a flow map in $\overline{\Omega^{\varepsilon}_{A,T}}$, and $v^{\varepsilon}_A = v^{\varepsilon}_A (x , t)$ be the velocity determined by the flow map $\tilde{x}^{\varepsilon}_A$, that is, for every $\xi_A \in \overline{\Omega_A (0)}$ and $0< t <T$,
\begin{equation*}
\begin{cases}
\tilde{x}^{\varepsilon}_{A} = \tilde{x}^{\varepsilon}_{A} ( \xi_{A} , t) = { }^t ( \tilde{x}^{A,\varepsilon}_1 (\xi_A , t) , \tilde{x}^{A,\varepsilon}_2 (\xi_{A}, t), \tilde{x}^{A,\varepsilon}_3 (\xi_{A},t) ),\\
v^{\varepsilon}_{A} = v^{\varepsilon}_{A} (x,t) = { }^t ( v^{A, \varepsilon}_1 (x,t) , v^{A,\varepsilon}_2 (x,t) , v^{A,\varepsilon}_3 (x,t) ),\\
\partial_t \tilde{x}^{\varepsilon}_{A} = v^{\varepsilon}_A ( \tilde{x}^{\varepsilon}_{A}(\xi_{A}, t) , t),\\
\tilde{x}^{\varepsilon}_{A} (\xi_{A}, 0 ) = \xi_{A}.
\end{cases} 
\end{equation*}
For each $i,j=1,2,3$,
\begin{equation*}
g_i^{A,\varepsilon} = g_i^{A,\varepsilon} ( \xi_{A} , t) := \frac{\partial \tilde{x}_{A}^{\varepsilon}}{\partial \xi^A_i} = { }^t \left( \frac{\partial \tilde{x}^{A,\varepsilon}_1 }{\partial \xi^{A}_i}, \frac{\partial \tilde{x}^{A,\varepsilon}_2 }{\partial \xi^{A}_i}, \frac{\partial \tilde{x}^{A,\varepsilon}_3 }{\partial \xi^{A}_i} \right),
\end{equation*}
\begin{equation*}
g^{A,\varepsilon}_{ij} = g^{A,\varepsilon}_{ij}(\xi_{A} , t) := g^{A,\varepsilon}_i \cdot g^{A,\varepsilon}_j,
\end{equation*}
and
\begin{equation*}
G^{\varepsilon}_{A} := G^{\varepsilon}_{A} (\xi_{A} , t) = {\rm{det}} (g_{ij}^{A,\varepsilon})_{3 \times 3}.
\end{equation*}

For each flow map $\tilde{x}^{\varepsilon}_B$ in $\overline{\Omega^\varepsilon_{B,T}}$, we define $g_i^{B, \varepsilon}$, $g^{B,\varepsilon}_{ij}$, $G^{\varepsilon}_B$ similarly to $g_i^{A, \varepsilon}$, $g^{A,\varepsilon}_{ij}$, $G^{\varepsilon}_A$.

\begin{remark}
Let $\sharp = A ,B$. Since ${ }^t (\nabla_{\xi_\sharp} \tilde{x}_\sharp) (\nabla_{\xi_\sharp} \tilde{x}_\sharp) = (g^\sharp_{ij})_{3 \times 3}$, ${ }^t (\nabla_{\xi_\sharp} \tilde{x}^\varepsilon_\sharp) (\nabla_{\xi_\sharp} \tilde{x}^\varepsilon_\sharp) = (g^{\sharp, \varepsilon}_{ij})_{3 \times 3}$, ${\rm{det}}({ }^t(\nabla_{\xi_\sharp} \tilde{x}_\sharp)) = {\rm{det}}(\nabla_{\xi_\sharp} \tilde{x}_\sharp)$, and ${\rm{det}}({ }^t(\nabla_{\xi_\sharp} \tilde{x}^\varepsilon_\sharp)) = {\rm{det}}(\nabla_{\xi_\sharp} \tilde{x}^\varepsilon_\sharp)$, we see that for $f \in C (\mathbb{R}^4)$,
\begin{align*}
\int_{\Omega_\sharp (t)} f (x ,t ) { \ }d x & = \int_{\Omega_\sharp (0)} f (\tilde{x}_\sharp (\xi_\sharp, t) , t) \sqrt{G_\sharp (\xi_\sharp , t)}{ \ }d \xi_\sharp,\\
\int_{\Omega^\varepsilon_\sharp (t)} f (x ,t ) { \ }d x & = \int_{\Omega_\sharp (0)} f (\tilde{x}^\varepsilon_\sharp (\xi_\sharp, t) , t) \sqrt{G^\varepsilon_\sharp (\xi_\sharp , t)}{ \ }d \xi_\sharp.
\end{align*}
We call $\sqrt{G_\sharp}$ and $\sqrt{G_\sharp^\varepsilon}$ the \emph{Riemannian metrics determined by the flow maps} $\tilde{x}_\sharp$ and $\tilde{x}_\sharp^\varepsilon$, respectively.
\end{remark}

Now we define the Riemannian metric determined by the flow map $\tilde{x}_S$ in $\overline{\Gamma_T}$. Let $\tilde{x}_S = \tilde{x}_S (\xi_S , t)$ be the flow map in $\overline{\Gamma_T}$, and $v_S = v_S (x,t)$ be the velocity determined by the flow map $\tilde{x}_S$, that is, for every $\xi_S \in \Gamma (0)$ and $0< t <T$,
\begin{equation*}
\begin{cases}
\tilde{x}_{S} = \tilde{x}_{S} ( \xi_{S} , t) = { }^t ( \tilde{x}^{S}_1 (\xi_S , t) , \tilde{x}^{S}_2 (\xi_{S}, t), \tilde{x}^{S}_3 (\xi_{S},t) ),\\
v_{S} = v_{S} (x,t) = { }^t ( v^{S}_1 (x,t) , v^{S}_2 (x,t) , v^{S}_3 (x,t) ),\\
\partial_t \tilde{x}_{S} = v_S ( \tilde{x}_{S}(\xi_{S}, t) , t),\\
\tilde{x}_{S} (\xi_{S}, 0 ) = \xi_{S}.
\end{cases} 
\end{equation*}
Since $\Gamma (0)$ is a closed Riemannian 2-dimensional manifold, there are $N \in \mathbb{N}$, $\Gamma_m \subset \Gamma (0)$, $\Phi_m = \Phi_m (X) \in [C^\infty ( \mathbb{R}^2)]^3$, $U_m \subset \mathbb{R}^2$, $\Psi_m = \Psi_m (\xi_S) \in C^\infty (\mathbb{R}^3)$ $(m=1,2, \cdots, N)$ such that
\begin{equation}\label{eq31}
\begin{cases}
{\displaystyle{\bigcup_{m=1}^N \Gamma_m = \Gamma (0)}},\\
\Gamma_m = \Phi_m (U_m),\\
\text{supp}\Psi_m \subset \Gamma_m,\\
\vert \vert  \Psi_m \vert \vert_{L^\infty} =1,\\
{\displaystyle{\sum_{m=1}^N \Psi_m = 1 \text{ on } \Gamma (0)}}.
\end{cases}
\end{equation}
This is a partition of unity. Fix $\xi_S \in \Gamma (0)$. Assume that $\xi_S \in \Gamma_m$ for some $m \in \{ 1, 2, \cdots, N \}$. Since we can write $\xi_S = \Phi_m (X)$ for some $X = { }^t (X_1, X_2) \in U_m \subset \mathbb{R}^2$, we set
\begin{equation*}
\hat{x}_S = \hat{x}_S(X,t) = { }^t (\hat{x}^S_1 , \hat{x}^S_2, \hat{x}^S_3) = \tilde{x}_S (\Phi_m(X), t) (= \tilde{x}_S(\xi_S ,t)) . 
\end{equation*}
Then
\begin{equation*}
\begin{cases}
\partial_t \hat{x}_S = \partial_t \hat{x}_S (X,t) = v_S ( \hat{x}_S (X,t) , t),\\
\hat{x}_S  \vert_{t =0} = \Phi_m (X) (= \xi_S).  
\end{cases}
\end{equation*}
Now we write
\begin{equation}\label{eq32}
\Phi : = \Phi_m { \ }\text{ if } \xi_S \in \Gamma_m.
\end{equation}
Then for each $\xi_S \in \Gamma (0)$ and $0 < t < T$,
\begin{equation*}
\begin{cases}
\partial_t \hat{x}_S = \partial_t \hat{x}_S (X , t) = v_S (\hat{x}_S (X,t) , t),\\
\hat{x}_S \vert_{t=0} = \Phi ( X) (= \xi_S ).
\end{cases}
\end{equation*}
We also call $\hat{x}_S$ a \emph{flow map} in $\overline{\Gamma_T}$. For the flow map $\hat{x}_S$ and $\alpha , \beta = 1,2$,
\begin{equation*}
\mathfrak{g}_\alpha = \mathfrak{g}_\alpha (X , t) := \frac{\partial \hat{x}_S}{\partial X_\alpha} = { }^t \left( \frac{\partial \hat{x}^S_1}{\partial X_\alpha},\frac{\partial \hat{x}^S_2}{\partial X_\alpha} , \frac{\partial \hat{x}^S_3}{\partial X_\alpha} \right),
\end{equation*}
\begin{equation*}
\mathfrak{g}_{\alpha \beta} = \mathfrak{g}_{\alpha \beta} (X,t) := \mathfrak{g}_\alpha \cdot \mathfrak{g}_\beta = \frac{\partial \hat{x}_S}{\partial X_\alpha} \cdot \frac{\partial \hat{x}_S}{\partial X_\beta},
\end{equation*}
and
\begin{equation*}
G_S = G_S (X,t) := \mathfrak{g}_{11} \mathfrak{g}_{22} - \mathfrak{g}_{12}\mathfrak{g}_{21}.
\end{equation*}

For each $- 1 < \varepsilon <1$, let $\tilde{x}^{\varepsilon}_S = \tilde{x}^{\varepsilon}_S (\xi_S , t)$ be a flow map in $\overline{\Gamma^{\varepsilon}_T}$, and $v^\varepsilon_S = v^\varepsilon_S (x,t)$ be the velocity determined by the flow map $\tilde{x}^\varepsilon_S$. Fix $\xi_S \in \Gamma (0)$. Assume that $\xi_S \in \Gamma_m$ for some $m \in \{ 1, 2, \cdots, N \}$. Since we can write $\xi_S = \Phi_m (X)$ for some $X = { }^t (X_1, X_2) \in U_m \subset \mathbb{R}^2$, we set
\begin{equation*}
\hat{x}_S^\varepsilon = \hat{x}_S^{\varepsilon}(X,t) = { }^t (\hat{x}_1^{S,\varepsilon} , \hat{x}_2^{S, \varepsilon} , \hat{x}_3^{S,\varepsilon})  := \tilde{x}_S^{\varepsilon} (\Phi_m(X), t) (= \tilde{x}_S^{\varepsilon }(\xi_S ,t)) . 
\end{equation*}
Then
\begin{equation*}
\begin{cases}
\partial_t \hat{x}_S^{\varepsilon} = \partial_t \hat{x}_S^{\varepsilon} (X,t) = v_S^{\varepsilon} ( \hat{x}_S^{\varepsilon}(X,t) , t),\\
\hat{x}_S^{\varepsilon} \vert_{t =0} = \Phi_m (X) (= \xi_S).  
\end{cases}
\end{equation*}
Now we use \eqref{eq32}. Then for each $\xi_S \in \Gamma (0)$ and $0 < t < T$,
\begin{equation*}
\begin{cases}
\partial_t \hat{x}_S^{\varepsilon} = \partial_t \hat{x}_S^{\varepsilon} (X , t) = v_S^{\varepsilon} (\hat{x}_S^{\varepsilon}(X,t) , t),\\
\hat{x}_S^{\varepsilon} \vert_{t=0} = \Phi ( X) (= \xi_S ).
\end{cases}
\end{equation*}
We also call $\hat{x}^\varepsilon_S$ a \emph{flow map} in $\overline{\Gamma^\varepsilon_T}$. For the flow map $\hat{x}^\varepsilon_S$ and $\alpha, \beta = 1,2$,
\begin{equation*}
\mathfrak{g}^\varepsilon_\alpha = \mathfrak{g}^\varepsilon_\alpha (X , t) := \frac{\partial \hat{x}^\varepsilon_S}{\partial X_\alpha} = { }^t \left( \frac{\partial \hat{x}^{S,\varepsilon}_1}{\partial X_\alpha},\frac{\partial \hat{x}^{S,\varepsilon}_2}{\partial X_\alpha} , \frac{\partial \hat{x}^{S,\varepsilon}_3}{\partial X_\alpha} \right),
\end{equation*}
\begin{equation*}
\mathfrak{g}^{\varepsilon}_{\alpha \beta} = \mathfrak{g}^\varepsilon_{\alpha \beta} (X,t) := \mathfrak{g}^\varepsilon_\alpha \cdot \mathfrak{g}^\varepsilon_\beta = \frac{\partial \hat{x}^\varepsilon_S}{\partial X_\alpha} \cdot \frac{\partial \hat{x}^\varepsilon_S}{\partial X_\beta},
\end{equation*}
and
\begin{equation*}
G^\varepsilon_S = G^\varepsilon_S (X,t) := \mathfrak{g}^\varepsilon_{11} \mathfrak{g}^\varepsilon_{22} - \mathfrak{g}^\varepsilon_{12}\mathfrak{g}^\varepsilon_{21}.
\end{equation*}
We call $\sqrt{G_S}$ the \emph{Riemannian metric determined by the flow map} $\hat{x}_S$, and $\sqrt{G^\varepsilon_S}$ the \emph{Riemannian metric determined by the flow map} $\hat{x}^\varepsilon_S$. In fact, the following lemma holds.
\begin{lemma}\label{lem32}
Let $\mathfrak{M}_S \subset \Gamma (0)$, $0<t <T$, and $-1 < \varepsilon <1$. Set
\begin{align*}
\Lambda_S (t) = \{ x \in \mathbb{R}^3; { \ }x = \tilde{x}_S (\xi_S , t) , { \ }\xi_S \in \mathfrak{M}_S \},\\
\Lambda_S^\varepsilon (t) = \{ x \in \mathbb{R}^3; { \ }x = \tilde{x}^\varepsilon_S (\xi_S , t) , { \ }\xi_S \in \mathfrak{M}_S \}.
\end{align*}
Then for each $f \in C (\mathbb{R}^3 \times \mathbb{R})$,
\begin{equation}\label{eq33}
\int_{\Gamma (t)} f (x,t) { \ }d \mathcal{H}^2_x = \sum_{m=1}^N \int_{U_m} \hat{\Psi}_m \hat{f} \sqrt{G_S (X,t)} { \ } d X,
\end{equation}
\begin{equation}\label{eq34}
\int_{\Gamma^\varepsilon (t)} f (x,t) { \ }d \mathcal{H}^2_x = \sum_{m=1}^N \int_{U_m} \hat{\Psi}_m \hat{f}_\varepsilon \sqrt{G^\varepsilon_S (X,t)} { \ } d X,
\end{equation}
\begin{equation}\label{eq35}
\int_{\Lambda_S (t)} f (x,t) { \ }d \mathcal{H}^2_x = \sum_{m=1}^N \int_{U_m} 1_{\mathfrak{M}_S \cap \Gamma_m}(\Phi_m(X))  \hat{\Psi}_m \hat{f} \sqrt{G_S (X,t)} { \ } d X,
\end{equation}
\begin{equation}\label{eq36}
\int_{\Lambda_S^\varepsilon (t)} f (x,t) { \ }d \mathcal{H}^2_x = \sum_{m=1}^N \int_{U_m} 1_{\mathfrak{M}_S \cap \Gamma_m}(\Phi_m(X)) \hat{\Psi}_m \hat{f}_\varepsilon \sqrt{G^\varepsilon_S (X,t)} { \ } d X.
\end{equation}
Here
\begin{equation}\label{eq37}
\begin{cases}
\hat{f} & = \hat{f} (X,t) := f (\tilde{x}_S(\Phi_m(X),t) ,t ),\\
\hat{f}_\varepsilon &= \hat{f}_\varepsilon (X,t) := f (\tilde{x}^\varepsilon_S(\Phi_m(X),t) ,t ),\\
\hat{\Psi}_m & =\hat{\Psi}_m (X) := \Psi_m (\Phi_m(X)).
\end{cases}
\end{equation}
\end{lemma}
\noindent We prove Lemma \ref{lem32} in Appendix. From now, we assume that the Riemannian metrics $\sqrt{G_A},\sqrt{G^\varepsilon_A}, \sqrt{G_B},\sqrt{G_B^\varepsilon}, \sqrt{G_S}, \sqrt{G_S^\varepsilon}$ are positive functions throughout this paper. From Lemma \ref{lem32}, \cite{KLG17}, \cite{K18}, \cite{K22}, and \cite{KS17}, we have the following two lemmas.
\begin{lemma}[Properties of the Riemannian metrics determined by flow maps]\label{lem33}{ \ }\\
Let $f = f (x,t) \in C^\infty (\mathbb{R}^4)$. Then
\begin{equation*}
\int_{\Omega_A (t)} f ({\rm{div}} v_A ) { \ }d x = \int_{\Omega_A (0) } f ( \tilde{x}_A(\xi_A , t) , t) \left( \frac{d}{d t} \sqrt{G_A} \right) { \ }d \xi_A,
\end{equation*}
\begin{equation*}
\int_{\Omega^{\varepsilon}_A (t)} f ({\rm{div}} v^{\varepsilon}_A ) { \ }d x = \int_{\Omega_A (0) } f ( \tilde{x}^\varepsilon_A(\xi_A , t) , t)  \left( \frac{d}{d t} \sqrt{G^\varepsilon_A} \right) { \ }d \xi_A,
\end{equation*}
\begin{equation*}
\int_{\Omega_A (t)} f ({\rm{div}} z_A ) { \ }d x = \int_{\Omega_A (0) } f ( \tilde{x}_A(\xi_A , t) , t) \left( \frac{d}{d \varepsilon} \bigg \vert_{\varepsilon = 0} \sqrt{G^{\varepsilon}_A} \right){ \ }d \xi_A,
\end{equation*}
\begin{equation*}
\int_{\Omega_B (t)} f ({\rm{div}} v_B ) { \ }d x = \int_{\Omega_B (0) } f ( \tilde{x}_B(\xi_B , t) , t)  \left( \frac{d}{d t} \sqrt{G_B} \right) { \ }d \xi_B,
\end{equation*}
\begin{equation*}
\int_{\Omega^{\varepsilon}_B (t)} f ({\rm{div}} v^{\varepsilon}_B ) { \ }d x = \int_{\Omega_B (0) } f ( \tilde{x}^\varepsilon_B (\xi_B , t) , t)  \left( \frac{d}{d t} \sqrt{G^\varepsilon_B} \right) { \ }d \xi_B,
\end{equation*}
\begin{equation*}
\int_{\Omega_B (t)} f ({\rm{div}} z_B ) { \ }d x = \int_{\Omega_B (0) } f ( \tilde{x}_B (\xi_B , t) , t)  \left( \frac{d}{d \varepsilon} \bigg \vert_{\varepsilon = 0} \sqrt{G^{\varepsilon}_B} \right){ \ }d \xi_B,
\end{equation*}
\begin{equation*}
\int_{\Gamma (t)} f ({\rm{div}}_\Gamma v_S ) { \ }d \mathcal{H}^2_x =\sum_{m=1}^N \int_{U_m } \hat{\Psi}_m \hat{f} \left( \frac{d}{dt} \sqrt{G_S} \right) { \ }d X, 
\end{equation*}
\begin{equation*}
\int_{\Gamma^{\varepsilon} (t)} f ({\rm{div}}_{\Gamma^\varepsilon} v^{\varepsilon}_S ) { \ }d \mathcal{H}^2_x = \sum_{m=1}^N \int_{U_m } \hat{\Psi}_m \hat{f}_\varepsilon \left( \frac{d}{dt} \sqrt{G^{\varepsilon}_S} \right) { \ }d X,
\end{equation*}
\begin{equation*}
\int_{\Gamma (t)} f ({\rm{div}}_{\Gamma} z_S ) { \ }d \mathcal{H}^2_x = \sum_{m=1}^N \int_{U_m } \hat{\Psi}_m \hat{f} \left( \frac{d}{d \varepsilon} \bigg \vert_{\varepsilon = 0} \sqrt{G^{\varepsilon}_S} \right) { \ }d X.
\end{equation*}
Here $(z_A, z_B, z_S)$ is a variation of $(\tilde{x}^\varepsilon_A, \tilde{x}^\varepsilon_B, \tilde{x}^\varepsilon_S)$ and $(\hat{f}, \hat{f}_\varepsilon, \hat{\Psi}_m)$ is defined by \eqref{eq37}.
\end{lemma}

\begin{lemma}[Representation formulas for the energy densities]\label{lem34}{ \ }\\
Let $\rho_0^A \in C^\infty ( \overline{\Omega_A (0)})$, $\rho_0^B \in C^\infty ( \overline{\Omega_B (0)})$, $\rho_0^S \in C^\infty ( \Gamma (0))$. Let $p_A, p_B , p_S \in C^1(\mathbb{R})$. Then two assertions hold:\\ 
\noindent $(\mathrm{i})$ Assume that $(\rho_A , \rho_B , \rho_S)$ satisfies \eqref{eq21} and \eqref{eq23}. Then
\begin{equation*}
\int_{\Omega_A (t)} \frac{1}{2}\rho_A  \vert  v_A \vert^2 { \ } d x = \int_{\Omega_A (0)} \frac{1}{2}\rho^A_0(\xi_A)  \vert  \partial_t \tilde{x}_A( \xi_A , t) \vert^2 { \ }d \xi_A,
\end{equation*}
\begin{equation*}
\int_{\Omega_B (t)} \frac{1}{2}\rho_B  \vert  v_B \vert^2 { \ } d x = \int_{\Omega_B (0)} \frac{1}{2}\rho^B_0(\xi_B)  \vert  \partial_t \tilde{x}_B( \xi_B , t) \vert^2 { \ }d \xi_B,
\end{equation*}
\begin{equation*}
\int_{\Gamma (t)} \frac{1}{2}\rho_S  \vert  v_S \vert^2 { \ } d \mathcal{H}^2_x = \sum_{m=1}^N \int_{U_m} \hat{\Psi}_m \frac{1}{2} \rho^S_0(\Phi_m(X))  \vert  \partial_t \hat{x}_S( X , t) \vert^2 { \ }d X,
\end{equation*}
\begin{equation*}
\int_{\Omega_A (t)} p_A (\rho_A) { \ } dx = \int_{\Omega_A (0)} p_A \left( \frac{\rho_0^A (\xi_A)}{\sqrt{G_A}} \right) \sqrt{G_A} { \ }d \xi_A,
\end{equation*}
\begin{equation*}
\int_{\Omega_B (t)} p_B (\rho_B) { \ } dx  = \int_{\Omega_B (0)} p_B \left( \frac{\rho_0^B (\xi_B)}{\sqrt{G_B}} \right) \sqrt{G_B} { \ }d \xi_B,
\end{equation*}
\begin{equation*}
\int_{\Gamma (t)} p_S (\rho_S) { \ } d \mathcal{H}^2_x = \sum_{m=1}^N \int_{U_m} \hat{\Psi}_m p_S \left( \frac{\rho_0^S(\Phi_m(X))}{\sqrt{G_S}} \right) \sqrt{G_S} { \ }d X.
\end{equation*}
Here $\hat{\Psi}_m$ is defined by \eqref{eq37}.\\
\noindent $(\mathrm{ii})$ Let $-1 < \varepsilon <1$. Assume that $(\rho^{\varepsilon}_A , \rho^{\varepsilon}_B , \rho^{\varepsilon}_S)$ satisfies \eqref{eq22} and \eqref{eq24}. Then
\begin{equation*}
\int_{\Omega^{\varepsilon}_A (t)} \frac{1}{2}\rho^{\varepsilon}_A  \vert  v^{\varepsilon}_A \vert^2 { \ } d x = \int_{\Omega^{\varepsilon}_A (0)} \frac{1}{2}\rho^A_0(\xi_A)  \vert  \partial_t \tilde{x}^{\varepsilon}_A( \xi_A , t) \vert^2 { \ }d \xi_A,
\end{equation*}
\begin{equation*}
\int_{\Omega^{\varepsilon}_B (t)} \frac{1}{2}\rho^{\varepsilon}_B  \vert  v^{\varepsilon}_B \vert^2 { \ } d x = \int_{\Omega_B (0)} \frac{1}{2}\rho^B_0(\xi_B)  \vert  \partial_t \tilde{x}^{\varepsilon}_B( \xi_B , t) \vert^2 { \ }d \xi_B,
\end{equation*}
\begin{equation*}
\int_{\Gamma^{\varepsilon} (t)} \frac{1}{2}\rho^{\varepsilon}_S  \vert  v^{\varepsilon}_S \vert^2 { \ } d \mathcal{H}^2_x = \sum_{m=1}^N \int_{U_m} \hat{\Psi}_m \frac{1}{2} \rho^S_0(\Phi_m(X))  \vert  \partial_t \hat{x}^{\varepsilon}_S( X , t) \vert^2 { \ }d X,
\end{equation*}
\begin{equation*}
\int_{\Omega^{\varepsilon}_A (t)} p_A (\rho^{\varepsilon}_A) { \ } dx = \int_{\Omega_A (0)} p_A \left( \frac{\rho_0^A (\xi_A)}{\sqrt{G^{\varepsilon}_A}} \right) \sqrt{G^{\varepsilon}_A} { \ }d \xi_A,
\end{equation*}
\begin{equation*}
\int_{\Omega^{\varepsilon}_B (t)} p_B (\rho^{\varepsilon}_B) { \ } dx = \int_{\Omega_B (0)} p_B \left( \frac{\rho_0^B (\xi_B)}{\sqrt{G^{\varepsilon}_B}} \right) \sqrt{G^{\varepsilon}_B} { \ }d \xi_B,
\end{equation*}
\begin{equation*}
\int_{\Gamma^{\varepsilon} (t)} p_S (\rho^{\varepsilon}_S) { \ } d \mathcal{H}^2_x = \sum_{m=1}^N \int_{U_m} \hat{\Psi}_m p_S \left( \frac{\rho_0^S(\Phi_m(X))}{\sqrt{G^{\varepsilon}_S}} \right) \sqrt{G^{\varepsilon}_S} { \ }d X.
\end{equation*}
Here $\hat{\Psi}_m$ is defined by \eqref{eq37}.
\end{lemma}

\section{Variations of action integral with respect to flow maps}\label{sect4}

Let us prove one of the main results. Let $\rho_0^A \in C^\infty ( \overline{\Omega_A(0)})$, $\rho_0^B \in C^\infty ( \overline{\Omega_B(0)})$, and $\rho_0^S \in C^\infty ( \Gamma (0))$. Assume that $(\rho_A , \rho_B , \rho_S)$ satisfy \eqref{eq21} and \eqref{eq23}. For each $- 1 < \varepsilon < 1$, assume that $( \rho^\varepsilon_A , \rho^\varepsilon_B , \rho^\varepsilon_S)$ satisfy \eqref{eq22} and \eqref{eq24}. Fix $p_A,p_B,p_S \in C^1 (\mathbb{R})$. Assume \eqref{eq25} and \eqref{eq26} hold.

We first study some properties of the variation $(z_A, z_B , z_S)$ of $(\tilde{x}^\varepsilon_A, \tilde{x}_B^\varepsilon , \tilde{x}_S^\varepsilon)$ (see Definition \ref{def26}).  Let $\xi_S \in \Gamma (0)$. Assume that $\xi_S = \Phi (X)$, where $\Phi$ is defined by \eqref{eq32}. Set
\begin{equation*}
\hat{y}_S = \hat{y}_S (X , t) = { }^t (\hat{y}^S_1, \hat{y}^S_2, \hat{y}^S_3) = \tilde{y}_S ( \Phi (X) , t) (= \tilde{y}_S (\xi_S , t  )).
\end{equation*}

\begin{lemma}[Properties of variations $(y_A,y_B,y_S, \hat{y}_S)$]\label{lem41} For all $\xi_A \in \overline{\Omega_A (0)}$, $\xi_B \in \overline{\Omega_B(0)}$, $\xi_S \in \Gamma (0)$,
\begin{align}
\tilde{y}_A (\xi_A, 0) & = \tilde{y}_A (\xi_A,T-) = { }^t (0, 0,0),\label{eq41}\\
\tilde{y}_B (\xi_B, 0) & = \tilde{y}_B (\xi_B , T-) = { }^t (0, 0,0),\label{eq42}\\
\tilde{y}_S (\xi_S, 0) & = \tilde{y}_X (\xi_S ,T-) = { }^t (0, 0,0),\label{eq43}\\
\hat{y}_S (X,0) & = \hat{y}_S (X,T-) = { }^t (0, 0,0),\label{eq44}
\end{align}
where $\xi_S = \Phi (X)$.
\end{lemma}

\begin{proof}[Proof of Lemma \ref{lem41}]
We only show \eqref{eq43}. Since $\tilde{x}_S^{\varepsilon} (\xi_S , 0) = \xi_S$ for every $- 1 < \varepsilon < 1$ and $\xi_S \in \Gamma (0)$, we find that
\begin{align*}
\tilde{y}_S (\xi_S, 0) = \frac{d \tilde{x}^\varepsilon_S}{d \varepsilon} ( \xi_S , 0) & = \lim_{h \to 0} \frac{\tilde{x}_S^{\varepsilon + h} (\xi_S , 0) - \tilde{x}_S^\varepsilon ( \xi_S, 0) }{h}\\
& = \frac{\xi_S - \xi_S }{h} = { }^t (0 , 0 , 0).
\end{align*}
By assumption \eqref{eq26}, we see that $\tilde{y}_S (\xi_S, T-)={ }^t (0,0,0)$. Thus, we have \eqref{eq43}. Therefore, the lemma follows.
\end{proof}

\begin{proposition}[Variations of our energies with respect to flow maps]\label{prop42}
\begin{equation}
\frac{d}{d \varepsilon} \bigg \vert_{\varepsilon = 0} \int_0^T \int_{\Omega^{\varepsilon}_A (t)} \frac{1}{2} \rho^{\varepsilon}_A  \vert  v^{\varepsilon}_A \vert^2 { \ }d x d t = - \int_0^T \int_{\Omega_A (t)} (\rho_A D_t^A v_A ) \cdot z_A  { \ }d x d t,\label{eq45}
\end{equation}
\begin{equation}
\frac{d}{d \varepsilon} \bigg \vert_{\varepsilon = 0} \int_0^T \int_{\Omega^{\varepsilon}_B (t)} \frac{1}{2} \rho^{\varepsilon}_B  \vert  v^{\varepsilon}_B \vert^2 { \ }d x d t = - \int_0^T \int_{\Omega_B (t)} (\rho_B D_t^B v_B ) \cdot z_B  { \ }d x d t,\label{eq46}
\end{equation}
\begin{equation}
\frac{d}{d \varepsilon} \bigg \vert_{\varepsilon = 0} \int_0^T \int_{\Gamma^{\varepsilon} (t)} \frac{1}{2} \rho^{\varepsilon}_S  \vert  v^{\varepsilon}_S \vert^2 { \ }d x d t = - \int_0^T \int_{\Gamma (t)} (\rho_S D_t^S v_S ) \cdot z_S  { \ }d \mathcal{H}_x^2 dt,\label{eq47}
\end{equation}
\begin{equation}
\frac{d}{d \varepsilon} \bigg \vert_{\varepsilon = 0} \int_{\Omega^{\varepsilon}_A (t)}p_A (\rho^{\varepsilon}_A ) { \ }d x = \int_{\Omega_A (t)} ( p_A (\rho_A) - \rho_A p'_A(\rho_A) ) {\rm{div}}z_A { \ }d x,\label{eq48}\\
\end{equation}
\begin{equation}
\frac{d}{d \varepsilon} \bigg \vert_{\varepsilon = 0} \int_{\Omega^{\varepsilon}_B (t)}p_B (\rho^{\varepsilon}_B ) { \ }d x = \int_{\Omega_B (t)} ( p_B (\rho_B) - \rho_B p'_B(\rho_B) ) {\rm{div}}z_B { \ }d x,\label{eq49}\\
\end{equation}
\begin{equation}
\frac{d}{d \varepsilon} \bigg \vert_{\varepsilon = 0} \int_{\Gamma^{\varepsilon} (t)}p_S (\rho^{\varepsilon}_S ) { \ }d \mathcal{H}_x^2 = \int_{\Gamma (t)} ( p_S (\rho_S) - \rho_S p'_S (\rho_S) ) {\rm{div}}_\Gamma z_S { \ }d \mathcal{H}_x^2, \label{eq4010}
\end{equation}
\begin{equation}
\frac{d}{d \varepsilon} \bigg \vert_{\varepsilon = 0} \int_{\Gamma^{\varepsilon} (t)} 1 { \ }d \mathcal{H}_x^2 = \int_{\Gamma (t)} {\rm{div}}_\Gamma z_S { \ }d \mathcal{H}_x^2. \label{eq4011}
\end{equation}
\end{proposition}

\begin{proof}[Proof of Proposition \ref{prop42}]
We first derive \eqref{eq45} and \eqref{eq46}. By Lemma \ref{lem34}, we check that
\begin{multline}\label{eq4012}
\frac{d}{d \varepsilon} \bigg \vert_{\varepsilon =0} \int_0^T \int_{\Omega^{\varepsilon}_A (t)} \frac{1}{2} \rho^{\varepsilon}_A  \vert v^{\varepsilon }_A  \vert^2 { \ }d x d t\\
 =\frac{d}{d \varepsilon} \bigg \vert_{\varepsilon =0} \int_0^T \int_{\Omega_A (0)} \frac{1}{2} \rho_0^A (\xi_A ) \partial_t \tilde{x}^{\varepsilon}_A (\xi_A,t) \cdot \partial_t \tilde{x}^{\varepsilon}_A (\xi_A ,t) { \ }d \xi_A dt\\
 = \int_0^T \int_{\Omega_A (0)} \rho_0^A (\xi_A ) \partial_t \tilde{x}_A (\xi_A , t) \cdot \partial_t \tilde{y}_A (\xi_A ,t) { \ }d \xi_A dt\\
= \int_0^T \int_{\Omega_A (0)} \rho_0^A (\xi_A ) v_A ( \tilde{x}_A (\xi_A , t ) , t ) \cdot \partial_t \tilde{y}_A (\xi_A ,t) { \ }d \xi_A dt.
\end{multline}
Here we used some properties of flow maps (see Definitions \ref{def23}, \ref{def25}, \ref{def26}). Using integration by parts with \eqref{eq41} and Lemma \ref{lem34}, we see that
\begin{multline*}
\text{(R.H.S.) of }\eqref{eq4012}\\ = - \int_0^T \int_{\Omega_A (0)} \frac{\rho_0^A (\xi_A )}{ \sqrt{G_A (\xi_A , t)}} \left[ \frac{d}{dt} v_A ( \tilde{x}_A (\xi_A , t ) , t ) \right] \cdot \tilde{y}_A (\xi_A ,t) \sqrt{G_A (\xi_A , t) }{ \ }d \xi_A dt\\
= - \int_0^T \int_{\Omega_A (t)}  \{ \rho_A D_t^A v_A \} (x,t) \cdot z_A (x,t){ \ }d x d t, 
\end{multline*}
which is \eqref{eq45}. Similarly, we have \eqref{eq46}.

Secondly, we show \eqref{eq47}. By Lemma \ref{lem34}, we check that
\begin{multline*}
\frac{d}{d \varepsilon} \bigg \vert_{\varepsilon =0} \int_0^T \int_{\Gamma^{\varepsilon} (t)} \frac{1}{2} \rho^{\varepsilon}_S  \vert v^{\varepsilon }_S  \vert^2 { \ }d x d t\\
 =\frac{d}{d \varepsilon} \bigg \vert_{\varepsilon =0} \int_0^T \sum_{m=1}^N \int_{U_m} \hat{\Psi}_m \frac{1}{2} \rho_0^S (\Phi_m(X) ) \partial_t \hat{x}^{\varepsilon}_S (X ,t) \cdot \partial_t \hat{x}^{\varepsilon}_S (X ,t) { \ }d X dt\\
 = \int_0^T \sum_{m=1}^N \int_{U_m} \hat{\Psi}_m \rho_0^S (\Phi_m(X)) \partial_t \hat{x}_S (X , t) \cdot \partial_t \hat{y}_S (X ,t) { \ }dX dt. 
\end{multline*}
Using integration by parts with \eqref{eq44} and Lemma \ref{lem34}, we see that
\begin{multline*}
= - \int_0^T \sum_{m=1}^N \int_{U_m} \hat{\Psi}_m \frac{\rho_0^S (\Phi_m(X))}{\sqrt{  G_S }} \left[ \frac{d}{dt} \hat{v}_S ( \hat{x}_S (X, t ) , t ) \right] \cdot \hat{y}_S \sqrt{ G_S} { \ }d X dt\\
= - \int_0^T \int_{\Gamma (t)}  \{ \rho_S D_t^S v_S \} (x,t) \cdot z_S (x,t){ \ }d x d t. 
\end{multline*}
Thus, we have \eqref{eq47}.

Thirdly, we derive \eqref{eq48} and \eqref{eq49}. By Lemmas \ref{lem33}, \ref{lem34}, we check that
\begin{multline*}
\frac{d}{d \varepsilon} \bigg \vert_{\varepsilon =0} \int_{\Omega^{\varepsilon}_A (t)} p_A \left( \rho^{\varepsilon}_A \right) { \ }d x = \frac{d}{d \varepsilon} \bigg \vert_{\varepsilon =0} \int_{\Omega_A (0)} p_A \left( \frac{\rho_0^A}{ \sqrt{G^\varepsilon_A}} \right) \sqrt{G^\varepsilon_A} { \ }d \xi_A\\
 = \int_{\Omega_A (0)} \left( - p'_A \left( \frac{\rho_0^A}{ \sqrt{G_A} } \right) \frac{\rho_0^A}{ \sqrt{G_A} } + p_A \left( \frac{\rho_0^A}{\sqrt{G_A}} \right) \right) \left( \frac{d}{d \varepsilon} \bigg \vert_{\varepsilon =0} \sqrt{G^\varepsilon_A} \right) { \ }d \xi_A\\
= \int_{\Omega_A (t)}( p_A (\rho_A) - \rho_A p'_A(\rho_A) ) ({\rm{div}} z_A){ \ }d x. 
\end{multline*}
Thus, we have \eqref{eq48}. Similarly, we see \eqref{eq49}.

Finally, we show \eqref{eq4010}. By Lemmas \ref{lem33}, \ref{lem34}, we observe that
\begin{multline*}
\frac{d}{d \varepsilon} \bigg \vert_{\varepsilon =0} \int_{\Gamma^{\varepsilon} (t)} p_S \left( \rho^{\varepsilon}_S \right) { \ }d \mathcal{H}^2_x = \frac{d}{d \varepsilon} \bigg \vert_{\varepsilon =0} \sum_{m=1}^N \int_{U_m} \hat{\Psi}_m p_S \left( \frac{\rho_0^S(\Phi_m(X))}{ \sqrt{G^\varepsilon_S}} \right) \sqrt{G^\varepsilon_S} { \ }d X\\
 = \sum_{m=1}^N \int_{U_m} \hat{\Psi}_m K_m(X,t)  \left ( \frac{d}{d \varepsilon} \bigg \vert_{\varepsilon =0} \sqrt{G_S^\varepsilon} \right){ \ }d X\\
= \int_{\Gamma (t)}( p_S (\rho_S) - \rho_S p'_S (\rho_S) ) ({\rm{div}}_\Gamma z_S){ \ }d \mathcal{H}^2_x. 
\end{multline*}
Here
\begin{equation*}
K_m(X,t) := - p'_S \left( \frac{\rho_0^S(\Phi_m(X))}{ \sqrt{G_S} } \right) \frac{\rho_0^S(\Phi_m(X))}{ \sqrt{G_S} } + p_S \left( \frac{\rho_0^S(\Phi_m(X))}{\sqrt{G_S}} \right).
\end{equation*}
Thus, we have \eqref{eq4010}. The derivation of \eqref{eq4011} is left for the reader. Therefore, Proposition \ref{prop42} is proved. 
\end{proof}

Let us attack Theorem \ref{thm28}.
\begin{proof}[Proof of Theorem \ref{thm28}]
From Proposition \ref{prop42}, we see that
\begin{multline*}
\frac{d}{d \varepsilon} \bigg \vert_{\varepsilon = 0} A [\tilde{x}^\varepsilon_A , \tilde{x}_B^{\varepsilon} , \tilde{x}^{\varepsilon}_S ] = - \int_0^T \int_{\Omega_A (t)} (\rho_A D_t^A v_A ) \cdot z_A  { \ }d x d t\\
 - \int_0^T \int_{\Omega_B (t)} (\rho_B D_t^B v_B ) \cdot z_B  { \ }d x d t - \int_0^T \int_{\Gamma (t)} (\rho_S D_t^S v_S ) \cdot z_S  { \ }d \mathcal{H}_x^2 d t\\
 + \int_0^T \int_{\Omega_A (t)} \mathfrak{p}_A {\rm{div}}z_A { \ }d x d t + \int_0^T \int_{\Omega_B (t)} \mathfrak{p}_B {\rm{div}}z_B { \ }d x d t\\ + \int_0^T \int_{\Gamma (t)} \mathfrak{p}_A {\rm{div}}_\Gamma z_S { \ }d \mathcal{H}_x^2 d t
\end{multline*}
where $(\mathfrak{p}_A, \mathfrak{p}_B, \mathfrak{p}_S)$ is defined by \eqref{eq14}. Using integration by parts (Lemma \ref{lem71}), we have \text{(R.H.S.)} of \eqref{eq27}. Therefore, Theorem \ref{thm28} is proved.
\end{proof}

\section{Mathematical modeling}\label{sect5}

Let us apply our energetic variational approaches to derive our multiphase flow systems \eqref{eq13}, \eqref{eq15}, and \eqref{eq18}. We suppose that $(v_A,v_B,v_S)$ satisfy \eqref{eq12}. From \eqref{eq12} we assume that  
\begin{equation}\label{eq51}
\begin{cases}
\varphi_B \cdot n_\Omega = 0 \text{ on }\partial \Omega,\\
\varphi_A \cdot n_\Gamma = \varphi_B \cdot n_\Gamma = \varphi_S \cdot n_\Gamma \text{ on }\Gamma (t).
\end{cases}
\end{equation}
Here $\varphi_A (x) := z_A (x,t)$, $\varphi_B (x) := z_B (x,t)$, and $\varphi_S (x) := z_S (x,t)$. Let $p_A,p_B,p_S \in C^1 (\mathbb{R})$. For the flow maps $(\tilde{x}_A, \tilde{x}_B, \tilde{x}_S)$ in $\overline{\Omega_T}$, 
\begin{multline}\label{eq52}
A_1 [\tilde{x}_A, \tilde{x}_B, \tilde{x}_S] := \int_0^T \int_{\Omega_A (t)} \left\{ \frac{1}{2} \rho_A  \vert v_A \vert^2 - p_A (\rho_A) \right\} (x,t) { \ }d x dt\\
 + \int_0^T \int_{\Omega_B (t)} \left\{ \frac{1}{2} \rho_B  \vert v_B \vert^2 - p_B (\rho_B) \right\} { \ }d x dt\\
 + \int_0^T \int_{\Gamma (t)} \left\{ \frac{1}{2} \rho_S  \vert v_S \vert^2 - p_S (\rho_S) \right\} { \ }d \mathcal{H}^2_x dt,
\end{multline}
\begin{multline}\label{eq53}
A_2 [\tilde{x}_A, \tilde{x}_B, \tilde{x}_S] := \int_0^T \int_{\Omega_A (t)} \left\{ \frac{1}{2} \rho_A  \vert v_A \vert^2 - p_A (\rho_A) \right\} (x,t) { \ }d x dt\\
 + \int_0^T \int_{\Omega_B (t)} \left\{ \frac{1}{2} \rho_B  \vert v_B \vert^2 - p_B (\rho_B) \right\} { \ }d x dt\\ + \int_0^T \int_{\Gamma (t)} \left\{ \frac{1}{2} \rho_S  \vert v_S \vert^2 \right\} { \ }d \mathcal{H}^2_x dt,
\end{multline}
\begin{multline}\label{eq54}
A_3 [\tilde{x}_A, \tilde{x}_B, \tilde{x}_S] := \int_0^T \int_{\Omega_A (t)} \frac{1}{2} \rho_A  \vert v_A \vert^2 { \ }d x dt
 + \int_0^T \int_{\Omega_B (t)} \frac{1}{2} \rho_B  \vert v_B \vert^2 { \ }d x dt\\
 + \int_0^T \int_{\Gamma (t)}  \frac{1}{2} \rho_S  \vert v_S \vert^2 { \ }d \mathcal{H}^2_x dt,
\end{multline}
and
\begin{multline}\label{eq55}
A_4 [\tilde{x}_A, \tilde{x}_B, \tilde{x}_S] := \int_0^T \int_{\Omega_A (t)} p_A (\rho_A) { \ }d x dt + \int_0^T \int_{\Omega_B (t)}  p_B (\rho_B) { \ }d x dt\\ + \int_0^T \int_{\Gamma (t)} p_S (\rho_S) { \ }d \mathcal{H}^2_x dt.
\end{multline}
Remark that the action integrals $A_1$, $A_2$ and $(A_3,A_4)$ correspond to systems \eqref{eq13}, \eqref{eq15}, and \eqref{eq16}, respectively. Now we study the variations of our action integrals with respect to flow maps. From Theorem \ref{thm28} and Proposition \ref{prop42}, we consider the following variational problems.
\begin{proposition}[Fundamental lemmas of calculus of variations]\label{prop51}Fix $0 < t <T$. Then the following three assertions holds:\\
\noindent $(\rm{i})$ Assume that for every $\varphi_A \in [C^\infty( \overline{\Omega_A (t)})]^3$,  $\varphi_B \in [C^\infty( \overline{\Omega_B (t)})]^3$,  $\varphi_S \in [C^\infty( \Gamma (t) )]^3$ satisfying \eqref{eq51}, 
\begin{multline*}
\int_{\Omega_A (t)} ( \rho_A D^A_t v_A + \nabla \mathfrak{p}_A) \cdot \varphi_A { \ } d x + \int_{\Omega_B (t)} (\rho_B D_t^B v_B + \nabla \mathfrak{p}_B) \cdot \varphi_B { \ } d x\\
+ \int_{\Gamma (t)} ( \rho_S D_t^S v_S + \nabla_\Gamma \mathfrak{p}_S + \mathfrak{p}_S H_\Gamma n_\Gamma - \mathfrak{p}_A n_\Gamma + \mathfrak{p}_B n_\Gamma  ) \cdot \varphi_S { \ } d \mathcal{H}^2_x = 0. 
\end{multline*}
Then
\begin{equation}\label{eq56}
\begin{cases}
\rho_A D^A_t v_A + \nabla \mathfrak{p}_A = 0 & \text{ in } \Omega_A (t),\\
\rho_B D^B_t v_B + \nabla \mathfrak{p}_B = 0 & \text{ in } \Omega_B(t),\\
\rho_S D_t^S v_S + \nabla_\Gamma \mathfrak{p}_S + \mathfrak{p}_S H_\Gamma n_\Gamma - \mathfrak{p}_A n_\Gamma + \mathfrak{p}_B n_\Gamma =0 & \text{ on } \Gamma (t).
\end{cases}
\end{equation}
Here $(\mathfrak{p}_A , \mathfrak{p}_B, \mathfrak{p}_S)$ is defined by \eqref{eq14}.\\
\noindent $(\rm{ii})$ Assume that for every $\varphi_A \in [C^\infty( \overline{\Omega_A (t)})]^3$,  $\varphi_B \in [C^\infty( \overline{\Omega_B (t)})]^3$,  $\varphi_S \in [C^\infty( \Gamma (t) )]^3$ satisfying \eqref{eq51} and ${\rm{div}}_\Gamma \varphi_S =0$ on $\Gamma (t)$, 
\begin{multline*}
\int_{\Omega_A (t)} ( \rho_A D^A_t v_A + \nabla \mathfrak{p}_A) \cdot \varphi_A { \ } d x + \int_{\Omega_B (t)} (\rho_B D_t^B v_B + \nabla \mathfrak{p}_B) \cdot \varphi_B { \ } d x\\
+ \int_{\Gamma (t)} ( \rho_S D_t^S v_S - \mathfrak{p}_A n_\Gamma + \mathfrak{p}_B n_\Gamma  ) \cdot \varphi_S { \ } d \mathcal{H}^2_x = 0. 
\end{multline*}
Then there exists $\Pi_S \in C^1(\Gamma (t))$ such that
\begin{equation}\label{eq57}
\begin{cases}
\rho_A D^A_t v_A + \nabla \mathfrak{p}_A = 0 & \text{ in } \Omega_A (t),\\
\rho_B D^B_t v_B + \nabla \mathfrak{p}_B = 0 & \text{ in } \Omega_B(t),\\
\rho_S D_t^S v_S + \nabla_\Gamma \Pi_S + \Pi_S H_\Gamma n_\Gamma - \mathfrak{p}_A n_\Gamma + \mathfrak{p}_B n_\Gamma =0 & \text{ on } \Gamma (t).
\end{cases}
\end{equation}
\noindent $(\rm{iii})$ Assume that for every $\varphi_A \in [C^\infty( \overline{\Omega_A (t)})]^3$, $\varphi_B \in [C^\infty( \overline{\Omega_B (t)})]^3$,  $\varphi_S \in [C^\infty( \Gamma (t) )]^3$ satisfying \eqref{eq51}, 
\begin{multline*}
\int_{\Omega_A (t)} \rho_A D^A_t v_A \cdot \varphi_A { \ } d x + \int_{\Omega_B (t)} \rho_B D_t^B v_B \cdot \varphi_B { \ } d x + \int_{\Gamma (t)} \rho_S D_t^S v_S \cdot P_\Gamma \varphi_S { \ } d \mathcal{H}^2_x\\
 = - \int_{\Omega_A (t)} \nabla \mathfrak{p}_A \cdot \varphi_A { \ } d x - \int_{\Omega_B (t)} \nabla \mathfrak{p}_B \cdot \varphi_B { \ } d x\\
 - \int_{\Gamma (t)} ( \nabla_\Gamma \mathfrak{p}_S + \mathfrak{p}_S H_\Gamma n_\Gamma - \mathfrak{p}_A n_\Gamma + \mathfrak{p}_B n_\Gamma  ) \cdot \varphi_S { \ } d \mathcal{H}^2_x. 
\end{multline*}
Then
\begin{equation}\label{eq58}
\begin{cases}
\rho_A D^A_t v_A + \nabla \mathfrak{p}_A = 0 & \text{ in } \Omega_A (t),\\
\rho_B D^B_t v_B + \nabla \mathfrak{p}_B = 0 & \text{ in } \Omega_B(t),\\
P_\Gamma \rho_S D_t^S v_S + \nabla_\Gamma \mathfrak{p}_S=0 & \text{ on } \Gamma (t),\\
\mathfrak{p}_S H_\Gamma n_\Gamma - \mathfrak{p}_A n_\Gamma + \mathfrak{p}_B n_\Gamma =0 & \text{ on } \Gamma (t).
\end{cases}
\end{equation}
\end{proposition}

\begin{proof}[Proof of Proposition \ref{prop51}]

Since the assertion $(\rm{i})$ is clear, we prove $(\rm{ii})$ and $(\rm{iii})$.

We first show $(\rm{ii})$. We now consider the case when $\varphi_B = { }^t (0,0,0)$ and $\varphi_S = { }^t (0,0,0)$, that is, for every $\varphi_A \in [C^\infty ( \overline{\Omega_A (t)} )]^3$ satisfying \eqref{eq51},
\begin{equation*}
\int_{\Omega_A (t)} ( \rho_A D^A_t v_A + \nabla \mathfrak{p}_A) \cdot \varphi_A { \ } d x = 0. 
\end{equation*}
This shows that
\begin{equation*}
\rho_A D^A_t v_A + \nabla \mathfrak{p}_A = 0 \text{ in } \Omega_A (t).
\end{equation*}
Similarly, we see that
\begin{equation*}
\rho_B D^B_t v_B + \nabla \mathfrak{p}_B = 0 \text{ in } \Omega_B(t).
\end{equation*}
Next we consider a general case, that is, for every $\varphi_S \in [C^\infty (\Gamma (t) )]^3$ satisfying ${\rm{div}}_\Gamma \varphi_S = 0$ on $\Gamma (t)$,
\begin{equation*}
\int_{\Gamma (t)} ( \rho_S D_t^S v_S - \mathfrak{p}_A n_\Gamma + \mathfrak{p}_B n_\Gamma  ) \cdot \varphi_S { \ } d \mathcal{H}^2_x = 0. 
\end{equation*}
Since ${\rm{div}}_\Gamma \varphi_S =0$, we apply the generalized Helmholtz-Weyl decomposition (Lemma \ref{lem72}) to find that there exists $\Pi_S \in C^1 (\Gamma (t))$ such that
\begin{equation*}
\rho_S D_t^S v_S  - \mathfrak{p}_A n_\Gamma + \mathfrak{p}_B n_\Gamma = - \nabla_\Gamma \Pi_S - \Pi_S H_\Gamma n_\Gamma \text{ on } \Gamma (t).
\end{equation*}
Thus, we see $(\rm{ii})$. 

Next we prove $(\rm{iii})$. We now consider the case when $\varphi_S = { }^t (0,0,0)$, that is, for every $\varphi_A \in [C^\infty ( \overline{\Omega_A (t)} )]^3$, $\varphi_B \in [C^\infty ( \overline{\Omega_B (t)} )]^3$ satisfying \eqref{eq51},
\begin{equation*}
\int_{\Omega_A (t)} ( \rho_A D^A_t v_A + \nabla \mathfrak{p}_A) \cdot \varphi_A { \ } d x + \int_{\Omega_B (t)} ( \rho_B D^B_t v_B + \nabla \mathfrak{p}_B) \cdot \varphi_B { \ } d x = 0. 
\end{equation*}
This shows that
\begin{align*}
\rho_A D^A_t v_A + \nabla \mathfrak{p}_A = 0 \text{ in } \Omega_A (t),\\
\rho_B D^B_t v_B + \nabla \mathfrak{p}_B = 0 \text{ in } \Omega_B (t).
\end{align*}
Next we consider a general case, that is, for every $\varphi_S \in [C^\infty (\Gamma (t) )]^3$,
\begin{equation*}
\int_{\Gamma (t)} ( P_\Gamma \rho_S D_t^S v_S +\nabla_\Gamma \mathfrak{p}_S + \mathfrak{p}_S H_\Gamma n_\Gamma - \mathfrak{p}_A n_\Gamma + \mathfrak{p}_B n_\Gamma  ) \cdot \varphi_S { \ } d \mathcal{H}^2_x = 0. 
\end{equation*}
Therefore, we see that
\begin{align*}
P_\Gamma \rho_S D_t^S v_S + \nabla_\Gamma \mathfrak{p}_S=0 & \text{ on } \Gamma (t),\\
\mathfrak{p}_S H_\Gamma n_\Gamma - \mathfrak{p}_A n_\Gamma + \mathfrak{p}_B n_\Gamma =0 & \text{ on } \Gamma (t).
\end{align*}
Thus, we see $(\rm{iii})$. Therefore, Proposition \ref{prop51} is proved.
\end{proof}

Let us derive our multiphase flow systems.
\subsection{Inviscid multiphase flow system with compressible surface flow}\label{subsec51}
Let us construct system \eqref{eq13}. Assume that $(v_A, v_B,v_S)$ satisfy \eqref{eq12}. Based on Proposition \ref{prop27}, we admit \eqref{eq21}. Set the action integral $A_1$ defined by \eqref{eq52}. We assume the following energetic variational principle: 
\begin{equation*}
\frac{\delta A_1}{\delta \tilde{x}} \bigg \vert_{z_B \cdot n_\Omega =0 \text{ on } \partial \Omega, { \ }z_A \cdot n_\Gamma = z_B \cdot n_\Gamma = z_S \cdot n_\Gamma \text{ on }\Gamma(t)} = 0.
\end{equation*}
Here $\tilde{x} = (\tilde{x}_A, \tilde{x}_B, \tilde{x}_S )$ and $(z_A, z_B,z_S)$ is a variation of $(\tilde{x}^\varepsilon_A, \tilde{x}^\varepsilon_B, \tilde{x}^\varepsilon_S )$. From assertion $(\rm{i})$ of Proposition \ref{prop51}, we derive \eqref{eq56}. Combining \eqref{eq21} and \eqref{eq56}, we therefore have system \eqref{eq13}.

\subsection{Inviscid multiphase flow system with incompressible surface flow}\label{subsec52}
Let us construct system \eqref{eq15}. Assume that $(v_A, v_B,v_S)$ satisfy \eqref{eq12}. Based on Proposition \ref{prop27}, we admit \eqref{eq21} and ${\rm{div}}_\Gamma v_S =0$ on $\Gamma_T$. Set the action integral $A_2$ defined by \eqref{eq53}. We assume the following energetic variational principle: 
\begin{equation*}
\frac{\delta A_2}{\delta \tilde{x}} \bigg \vert_{z_B \cdot n_\Omega =0 \text{ on } \partial \Omega{ \ }z_A \cdot n_\Gamma = z_B \cdot n_\Gamma = z_S \cdot n_\Gamma \text{ on }\Gamma(t), { \ }{\rm{div}}_\Gamma z_S =0 } = 0.
\end{equation*}
Here $\tilde{x} = (\tilde{x}_A, \tilde{x}_B, \tilde{x}_S )$ and $(z_A, z_B,z_S)$ is a variation of $(\tilde{x}^\varepsilon_A, \tilde{x}^\varepsilon_B, \tilde{x}^\varepsilon_S )$. From assertion $(\rm{ii})$ of Proposition \ref{prop51}, we derive \eqref{eq57}. Combining \eqref{eq21}, ${\rm{div}}_\Gamma v_S =0$, and \eqref{eq57}, we therefore have system \eqref{eq15}.

\subsection{Inviscid multiphase flow system with tangential surface flow}\label{subsec53}
Let us construct system \eqref{eq13}. Assume that $(v_A, v_B,v_S)$ satisfy \eqref{eq12}. Based on Proposition \ref{prop27}, we admit \eqref{eq21}. Set the action integrals $A_3$ and $A_4$ defined by \eqref{eq54} and \eqref{eq55}, respectively. We assume the following energetic variational principle: 
\begin{equation*}
\frac{\delta A_3}{\delta \tilde{x}} \bigg \vert_{ z_S \cdot n_\Gamma =0 \text{ on }\Gamma(t)} = \frac{\delta A_4}{\delta \tilde{x}} \bigg \vert_{z_B \cdot n_\Omega =0 \text{ on } \partial \Omega, { \ }z_A \cdot n_\Gamma = z_B \cdot n_\Gamma = z_S \cdot n_\Gamma \text{ on }\Gamma(t)}.
\end{equation*}
Here $\tilde{x} = (\tilde{x}_A, \tilde{x}_B, \tilde{x}_S )$ and $(z_A, z_B,z_S)$ is a variation of $(\tilde{x}^\varepsilon_A, \tilde{x}^\varepsilon_B, \tilde{x}^\varepsilon_S )$. From assertion $(\rm{iii})$ of Proposition \ref{prop51}, we derive \eqref{eq58}. Combining \eqref{eq21} and \eqref{eq58}, we therefore have system \eqref{eq16}.

\section{Conservation and energy laws}\label{sect6}
Let us study the conservation and energy laws of our systems.
\begin{proof}[Proof of Theorem \ref{thm29}]
We only show $(\rm{i})$. We first derive \eqref{eq18}. From Proposition \ref{prop27}, we find that for $0 <t <T$,
\begin{equation*}
\frac{d}{d t } \left( \int_{\Omega_A (t )} \rho_A (x,t) { \ }d x + \int_{\Omega_B (t)} \rho_B (x ,t ) { \ }d x + \int_{\Gamma (t)} \rho_S (x, t) { \ }d \mathcal{H}^2_x  \right) = 0. 
\end{equation*}
Integrating with respect to $t$, we have \eqref{eq18}. 

Secondly, we derive \eqref{eq19} and \eqref{eq1010}. Applying system \eqref{eq13} and Lemma \ref{lem33}, we check that
\begin{multline}\label{eq61}
\frac{d}{d t } \left( \int_{\Omega_A (t)} \rho_A v_A { \ }d x + \int_{\Omega_B (t)} \rho_B v_B { \ }d x + \int_{\Gamma (t)} \rho_S v_S { \ }d \mathcal{H}^2_x  \right)\\
= \int_{\Omega_A (t)} \rho_A D_t^A v_A { \ }d x + \int_{\Omega_B (t)} \rho_B D_t^B v_B { \ }d x + \int_{\Gamma (t)} \rho_S D_t^S v_S{ \ }d \mathcal{H}^2_x\\
= \int_{\Omega_A (t)} (-\nabla \mathfrak{p}_A) { \ }d x + \int_{\Gamma_B (t)} (- \nabla \mathfrak{p}_B) { \ }d x\\
 + \int_{\Gamma (t)}(- \nabla_\Gamma \mathfrak{p}_S - \mathfrak{p}_S H_\Gamma n_\Gamma + \mathfrak{p}_A n_\Gamma - \mathfrak{p}_B n_\Gamma ) { \ }d \mathcal{H}^2_x.
\end{multline}
Using the divergence theorems (Lemma \ref{lem71}), we find that
\begin{multline*}
\text{(R.H.S.) of }\eqref{eq61} = \int_{\Omega_A (t)} {\rm{div}}( - \mathfrak{p}_A I_{3 \times 3} ) { \ }d x + \int_{\Gamma_B (t)}{\rm{div}} (- \mathfrak{p}_B I_{3 \times 3} ) { \ }d x\\
 + \int_{\Gamma (t)}(- \nabla_\Gamma \mathfrak{p}_S - \mathfrak{p}_S H_\Gamma n_\Gamma + \mathfrak{p}_A n_\Gamma - \mathfrak{p}_B n_\Gamma ) { \ }d \mathcal{H}^2_x = - \int_{\partial \Omega} \mathfrak{p}_B n_\Omega { \ } d \mathcal{H}_x^2 .
\end{multline*}
Note that $\int_{\Gamma (t)}\partial_j^\Gamma \mathfrak{p} { \ }d \mathcal{H}^2_x = - \int_{\Gamma (t)} \mathfrak{p}_S H_\Gamma n_j { \ }d \mathcal{H}^2_x $. Integrating with respect to $t$, we have \eqref{eq19}. In the same manner, we observe that
\begin{multline}\label{eq62}
\frac{d}{d t } \left( \int_{\Omega_A (t)} \frac{1}{2} \rho_A  \vert  v_A \vert^2 { \ }d x + \int_{\Omega_B (t)} \frac{1}{2} \rho_B  \vert v_B \vert^2 { \ }d x + \int_{\Gamma (t)} \frac{1}{2} \rho_S  \vert v_S \vert^2 { \ }d \mathcal{H}^2_x  \right)\\
= \int_{\Omega_A (t)} \rho_A D_t^A v_A \cdot v_A { \ }d x + \int_{\Omega_B (t)} \rho_B D_t^B v_B \cdot v_B { \ }d x + \int_{\Gamma (t)} \rho_S D_t^S v_S \cdot v_S { \ }d \mathcal{H}^2_x\\
= \int_{\Omega_A (t)} (-\nabla \mathfrak{p}_A) \cdot v_A { \ }d x + \int_{\Omega_B (t)} (- \nabla \mathfrak{p}_B) \cdot v_B { \ }d x\\
 + \int_{\Gamma (t)} (- \nabla_\Gamma \mathfrak{p}_S - \mathfrak{p}_S H_\Gamma n_\Gamma + \mathfrak{p}_A n_\Gamma - \mathfrak{p}_B n_\Gamma ) \cdot v_S { \ }d \mathcal{H}^2_x\\
= \int_{\Omega_A (t)} ({\rm{div}} v_A) \mathfrak{p}_A { \ }d x + \int_{\Omega_B (t)} ({\rm{div}} v_B) \mathfrak{p}_B { \ }d x + \int_{\Gamma (t)} ({\rm{div}}_\Gamma v_S) \mathfrak{p}_S { \ }d \mathcal{H}_x^2.
\end{multline}
Note that $v_B \cdot n_\Omega =0$. Integrating with respect to $t$, we have \eqref{eq1010}.

Finally, we derive \eqref{eq1011}. Using Lemma \ref{lem33} and \eqref{eq14}, we see that 
\begin{multline}\label{eq63}
\frac{d}{dt}\left(  \int_{\Omega_A (t)} p_A (\rho_A) { \ } d x +  \int_{\Omega_B (t)} p_B (\rho_B) { \ } d x + \int_{\Gamma (t)} p_S (\rho_S) { \ } d \mathcal{H}^2_x \right)\\
= - \int_{\Omega_A (t)} \mathfrak{p}_A ({\rm{div}}v_A) { \ }d x  - \int_{\Omega_B (t)} \mathfrak{p}_B ({\rm{div}}v_B) { \ }d x - \int_{\Gamma (t)} \mathfrak{p}_S ({\rm{div}}_\Gamma v_S) { \ }d \mathcal{H}^2_x. 
\end{multline}
Here we used the fact that
\begin{multline*}
\frac{d}{dt} \int_{\Omega_A (t)} p_A (\rho_A) { \ } d x = \int_{\Omega_A (t)} \{ (D_t^A \rho_A) p'_A (\rho_A) + p_A (\rho_A ) ({\rm{div}} v_A) \}{ \ }d x\\
= \int_{\Omega_A (t)} \{ (-{\rm{div}} v_A) \rho_A p'_A (\rho_A) + p_A (\rho_A ) ({\rm{div}} v_A) \}{ \ }d x = - \int_{\Omega_A (t)} \mathfrak{p}_A ({\rm{div}}v_A) { \ }d x. 
\end{multline*}
From \eqref{eq62} and \eqref{eq63} we see \eqref{eq1011}. Therefore, Theorem \ref{thm29} is proved.
\end{proof}

\section{Appendix}

In Appendix, we prove Lemma \ref{lem32}, and introduce several formulas for integration by parts and a generalized Helmholtz-Weyl Decomposition on a closed surface.

\begin{proof}[Proof of Lemma \ref{lem32}]
Let $N$, $\Gamma_m$, $\Phi_m$, $U_m$, $\Psi_m$ be the symbols appearing in \eqref{eq31}. We only show \eqref{eq33} and \eqref{eq35}. We first attack \eqref{eq33}. Fix $0<t<T$. For each $m \in \{ 1,2, \cdots , N \}$,
\begin{equation*}
\Gamma_m (t) = \{ x \in \mathbb{R}^3; { \ }x = \tilde{x}_S (\Phi_m (X), t) , { \ } X \in U_m \}.
\end{equation*}
Since the mapping $\tilde{x}_S (\cdot , t ) : \Gamma (0) \to \Gamma (t)$ is bijective, we see that
\begin{equation*}
\Gamma (t) = \bigcup_{m=1}^N \Gamma_m (t) 
\end{equation*}
and that the inverse mapping $\eta_S = \eta_S (x,t) $ of $\tilde{x}_S (\xi_S ,t)$ exists, that is,
\begin{align*}
\xi_S = \eta_S (x , t),\\
\eta_S (\tilde{x}_S(\xi_S, t) ,t )  = \xi_S,\\
\tilde{x}_S (\eta_S (x,t) ,t) = x.
\end{align*}
Now we set
\begin{equation*}
\breve{\Psi}_m (x,t) = \Psi_m (\eta_S (x ,t )).
\end{equation*}
By definition, we find that
\begin{align*}
\text{supp}\breve{\Psi}_m \subset \Gamma_m (t),\\
\vert \vert  \breve{\Psi}_m \vert \vert_{L^\infty} =1,\\
\sum_{m=1}^N \breve{\Psi}_m = 1 \text{ on } \Gamma (t).
\end{align*}
Therefore, we see that $\breve{\Psi}_m$ is a partition of unity. From $\breve{\Psi}_m (x,t ) = \Psi_m( \eta_S (x,t)) = \Psi_m (\eta_S (\tilde{x}_S (\xi_S,t) ,t )) = \Psi_m (\xi_S ,t) = \Psi_m (\Phi_m (X))$, we check that 
\begin{multline*}
\int_{\Gamma (t)} f (x,t) { \ }d \mathcal{H}^2_x = \sum_{m=1}^N \int_{\Gamma_m(t)} \breve{\Psi}_m(x,t) f (x,t) { \ } d \mathcal{H}_x^2\\
= \sum_{m=1}^N \int_{\Gamma_m(t)} \Psi_m( \eta_S (x,t)) f (x,t) { \ } d \mathcal{H}_x^2\\
= \sum_{m=1}^N \int_{U_m} \Psi_m (\Phi_m (X)) f (\tilde{x}_S(\Phi_m (X),t),t)  {\rm{det}} \left({ }^t(\nabla_X \tilde{x}_S)(\nabla_X \tilde{x}_S) \right)  { \ } d X\\
= \sum_{m=1}^N \int_{U_m} \Psi_m (\Phi_m (X)) f (\tilde{x}_S(\Phi_m (X),t),t) \sqrt{G_S (X,t)} { \ } d X,
\end{multline*}
which is \eqref{eq33}.

Next we prove \eqref{eq35}. Since $\Lambda_S (t) \subset \Gamma (t)$, there is $V_m \subset U_m$ such that
\begin{equation*}
\mathfrak{M}_S \cap \Gamma_m = \{ \xi_S \in \mathbb{R}^3;{ \ }\xi_S = \Phi_m (X) , { \ }X \in V_m \}.
\end{equation*}
It is clear that
\begin{equation*}
\Lambda_S (t) \cap \Gamma_m(t) = \{ x \in \mathbb{R}^3;{ \ }x = \tilde{x}_S (\Phi_m(X) ,t), X \in V_m \}.
\end{equation*}
Therefore, we see that for $X \in U_m$
\begin{equation*}
1_{V_m} (X) = 1_{\mathfrak{M}_S \cap \Gamma_m} (\Phi_m(X)) = 1_{\Lambda_S(t) \cap \Gamma_m(t)} (\tilde{x}_S (\Phi_m(X),t)).
\end{equation*}
By the previous argument to derive \eqref{eq33}, we check that
\begin{multline*}
\int_{\Lambda_S (t)} f (x,t) { \ }d \mathcal{H}^2_x = \int_{\Gamma (t)} 1_{\Lambda_S (t)}(x) f (x,t) { \ }d \mathcal{H}^2_x\\ 
= \sum_{m=1}^N \int_{\Gamma_m(t)}  1_{\Lambda_S (t) \cap \Gamma_m(t)}(x) \breve{\Psi}_m(x,t) f (x,t) { \ } d \mathcal{H}_x^2\\
= \sum_{m=1}^N \int_{\Gamma_m(t)} 1_{\Lambda_S (t) \cap \Gamma_m(t)}(x) \Psi_m( \eta_S (x,t)) f (x,t) { \ } d \mathcal{H}_x^2\\
= \sum_{m=1}^N \int_{U_m} 1_{\mathfrak{M}_S \cap \Gamma_m}(\Phi_m(X)) \Psi_m (\Phi_m (X)) f (\tilde{x}_S(\Phi_m (X),t),t) \sqrt{G_S (X,t)} { \ } d X.
\end{multline*}
Thus, we see \eqref{eq35}. Therefore, Lemma \ref{lem32} is proved.
\end{proof}

\begin{lemma}[Formulas for integration by parts]\label{lem71}{ \ }\\
Fix $0 \leq t < T$ and $j=1,2,3$. Then the following three assertions hold:\\
\noindent $(\rm{i})$ For every $f_A,g_A \in C^1 ( \overline{\Omega_A (t)})$, $F_A \in [ C^1 ( \overline{\Omega_A (t)})]^3$,
\begin{equation*}
\int_{\Omega_A (t)} {\rm{div} } F_A { \ }d x = \int_{\Gamma (t)} F_A \cdot n_\Gamma { \ } d \mathcal{H}^2_x,
\end{equation*}
\begin{equation*}
\int_{\Omega_A (t)} (\partial_j f_A ) g_A { \ } d x = - \int_{\Omega_A (t)} f_A (\partial_j g_A ) { \ } d x + \int_{\Gamma (t)} f_A  g_A n_j^\Gamma { \ } d \mathcal{H}^2_x.
\end{equation*}
\noindent $(\rm{ii})$ For every $f_B,g_B \in C^1 ( \overline{\Omega_B (t)})$, $F_B \in [ C^1 ( \overline{\Omega_B (t)})]^3$,
\begin{equation*}
\int_{\Omega_B (t)} {\rm{div} } F_B { \ }d x = \int_{\partial \Omega} F_B \cdot n_\Omega { \ }d \mathcal{H}^2_x - \int_{\Gamma (t)} F_B \cdot n_\Gamma { \ }d \mathcal{H}^2_x,
\end{equation*}
\begin{multline*}
\int_{\Omega_B (t)} (\partial_j f_B ) g_B { \ } d x\\
 = - \int_{\Omega_B (t)} f_B (\partial_j g_B ) { \ } d x + \int_{\partial \Omega} f_B  g_B n_j^\Omega { \ } d \mathcal{H}^2_x - \int_{\Gamma (t)} f_B  g_B n_j^\Gamma { \ } d \mathcal{H}^2_x.
\end{multline*}
\noindent $(\rm{iii})$ For every $f_S,g_S \in C^1 ( \Gamma(t))$, $F_S \in [C^1 (\Gamma (t))]^3$,
\begin{align*}
\int_{\Gamma (t)} {\rm{div}}_\Gamma F_S { \ } d \mathcal{H}^2_x & = - \int_{\Gamma (t)} H_\Gamma ( F_S \cdot n_\Gamma) { \ } d \mathcal{H}^2_x,\\
\int_{\Gamma (t)} (\partial^\Gamma_j f_S ) g_S { \ } d \mathcal{H}^2_x & = - \int_{\Gamma (t)} f_S (\partial^\Gamma_j g_S ) { \ } d \mathcal{H}^2_x - \int_{\Gamma (t)} H_\Gamma f_S  g_S n_j^\Gamma { \ } d \mathcal{H}^2_x.
\end{align*}
\end{lemma}
\noindent The proof of surface divergence theorem (the assertion $(\rm{iii})$ of Lemma \ref{lem71}) can be founded in Simon \cite{Sim83} and Koba \cite{K19}.

\begin{lemma}[Generalized Helmholtz-Weyl decomposition]\label{lem72}
Let $\Gamma_*$ be a smooth closed 2-dimensional surface in $\mathbb{R}^3$. Set
\begin{equation*}
C^\infty_{{\rm{div}}_{\Gamma_*} }( \Gamma_*) = \{ \varphi \in [ C^\infty (\Gamma_*) ]^3;  { \ }{\rm{div}}_{\Gamma_*} \varphi = 0 \}.
\end{equation*}
Let $F_* \in [C^1(\Gamma_*)]^3$. Assume that for each $\varphi_* \in C^\infty_{{\rm{div}}_{\Gamma_*} }( \Gamma_*)$
\begin{equation*}
\int_{\Gamma_*} F_* \cdot \varphi_* { \ }d \mathcal{H}^2_x = 0.
\end{equation*}
Then there is $\Pi_* \in C^1 (\Gamma_*)$ such that $F_* = \nabla_{\Gamma_*} \Pi_* + \Pi_* H_{\Gamma_*} n_{\Gamma_*}$. Here $n_{\Gamma_*} = n_{\Gamma_*}(x)$ is the unit outer normal vector at $x \in \Gamma_*$ and $H_{\Gamma_*} = - {\rm{div}}_{\Gamma_*} n_*$.
\end{lemma}
\noindent The proof of Lemma \ref{lem72} can be founded in Koba-Liu-Giga \cite{KLG17}.

\end{document}

%% file: pic1moto.tex
%WinTpicVersion4.32a
{\unitlength 0.1in%
\begin{picture}(44.7600,22.9000)(11.3800,-36.3000)%
% STR 2 0 3 0 Black White  
% 4 2160 1960 2160 2060 2 0 0 0
% $n_\Gamma$
\put(21.6000,-20.6000){\makebox(0,0)[lb]{$n_\Gamma$}}%
% CIRCLE 2 0 3 0 Black White  
% 4 2120 2530 2890 3140 2890 3140 2890 3140
% 
\special{pn 8}%
\special{ar 2120 2530 982 982 0.0000000 6.2831853}%
% CIRCLE 2 0 3 0 Black White  
% 4 4632 2510 5402 3120 5402 3120 5402 3120
% 
\special{pn 8}%
\special{ar 4632 2510 982 982 0.0000000 6.2831853}%
% ELLIPSE 2 0 3 0 Black White  
% 4 2120 2520 2680 2920 2680 2920 2680 2920
% 
\special{pn 8}%
\special{ar 2120 2520 560 400 0.0000000 6.2831853}%
% SPLINE 2 0 3 0 Black White  
% 41 5162 2320 5176 2383 5177 2449 5162 2517 5135 2583 5094 2648 5040 2710 4977 2766 4904 2816 4825 2858 4738 2892 4650 2917 4559 2930 4471 2933 4386 2926 4305 2907 4233 2880 4170 2844 4115 2797 4075 2745 4048 2685 4033 2622 4032 2556 4047 2488 4075 2421 4115 2357 4169 2295 4233 2239 4305 2189 4386 2147 4471 2113 4560 2088 4650 2076 4738 2072 4824 2079 4904 2098 4977 2125 5041 2162 5094 2207 5135 2260 5162 2320
% 
\special{pn 8}%
\special{pa 5162 2320}%
\special{pa 5170 2351}%
\special{pa 5176 2382}%
\special{pa 5178 2414}%
\special{pa 5177 2446}%
\special{pa 5172 2478}%
\special{pa 5164 2509}%
\special{pa 5154 2540}%
\special{pa 5142 2569}%
\special{pa 5127 2598}%
\special{pa 5110 2625}%
\special{pa 5092 2651}%
\special{pa 5072 2676}%
\special{pa 5050 2700}%
\special{pa 5028 2722}%
\special{pa 5004 2744}%
\special{pa 4979 2764}%
\special{pa 4954 2783}%
\special{pa 4927 2801}%
\special{pa 4900 2818}%
\special{pa 4872 2834}%
\special{pa 4844 2849}%
\special{pa 4815 2863}%
\special{pa 4785 2875}%
\special{pa 4755 2886}%
\special{pa 4725 2896}%
\special{pa 4694 2906}%
\special{pa 4663 2914}%
\special{pa 4632 2921}%
\special{pa 4600 2926}%
\special{pa 4569 2929}%
\special{pa 4537 2932}%
\special{pa 4505 2933}%
\special{pa 4473 2933}%
\special{pa 4441 2932}%
\special{pa 4409 2929}%
\special{pa 4377 2925}%
\special{pa 4346 2918}%
\special{pa 4315 2910}%
\special{pa 4284 2900}%
\special{pa 4254 2889}%
\special{pa 4225 2876}%
\special{pa 4197 2861}%
\special{pa 4170 2844}%
\special{pa 4144 2824}%
\special{pa 4121 2803}%
\special{pa 4099 2779}%
\special{pa 4080 2753}%
\special{pa 4064 2725}%
\special{pa 4052 2696}%
\special{pa 4042 2665}%
\special{pa 4035 2634}%
\special{pa 4031 2602}%
\special{pa 4031 2570}%
\special{pa 4034 2538}%
\special{pa 4041 2507}%
\special{pa 4051 2476}%
\special{pa 4063 2447}%
\special{pa 4077 2418}%
\special{pa 4093 2390}%
\special{pa 4110 2363}%
\special{pa 4130 2338}%
\special{pa 4151 2314}%
\special{pa 4173 2291}%
\special{pa 4197 2269}%
\special{pa 4221 2248}%
\special{pa 4246 2229}%
\special{pa 4272 2210}%
\special{pa 4299 2193}%
\special{pa 4327 2176}%
\special{pa 4355 2162}%
\special{pa 4384 2148}%
\special{pa 4414 2135}%
\special{pa 4443 2123}%
\special{pa 4473 2112}%
\special{pa 4504 2102}%
\special{pa 4535 2094}%
\special{pa 4566 2087}%
\special{pa 4597 2082}%
\special{pa 4629 2078}%
\special{pa 4661 2075}%
\special{pa 4693 2073}%
\special{pa 4725 2072}%
\special{pa 4757 2072}%
\special{pa 4821 2078}%
\special{pa 4852 2084}%
\special{pa 4883 2092}%
\special{pa 4914 2101}%
\special{pa 4944 2111}%
\special{pa 4974 2124}%
\special{pa 5003 2138}%
\special{pa 5030 2155}%
\special{pa 5056 2173}%
\special{pa 5081 2194}%
\special{pa 5103 2217}%
\special{pa 5123 2242}%
\special{pa 5140 2269}%
\special{pa 5153 2298}%
\special{pa 5162 2320}%
\special{fp}%
% STR 2 0 3 0 Black White  
% 4 1260 3650 1260 3750 2 0 0 0
% $\Omega = \Omega_A (t) \cup \Gamma (t) \cup \Omega_B (t)$
\put(12.6000,-37.5000){\makebox(0,0)[lb]{$\Omega = \Omega_A (t) \cup \Gamma (t) \cup \Omega_B (t)$}}%
% STR 2 0 3 0 Black White  
% 4 3860 3660 3860 3760 2 0 0 0
% $\Omega = \Omega_A^\varepsilon (t) \cup \Gamma^\varepsilon (t) \cup \Omega_B^\varepsilon (t)$
\put(38.6000,-37.6000){\makebox(0,0)[lb]{$\Omega = \Omega_A^\varepsilon (t) \cup \Gamma^\varepsilon (t) \cup \Omega_B^\varepsilon (t)$}}%
% STR 2 0 3 0 Black White  
% 4 1920 2500 1920 2600 2 0 0 0
% $\Omega_A (t)$
\put(19.2000,-26.0000){\makebox(0,0)[lb]{$\Omega_A (t)$}}%
% STR 2 0 3 0 Black White  
% 4 4410 2520 4410 2620 2 0 0 0
% $\Omega^\varepsilon_A (t)$
\put(44.1000,-26.2000){\makebox(0,0)[lb]{$\Omega^\varepsilon_A (t)$}}%
% STR 2 0 3 0 Black White  
% 4 1930 3140 1930 3240 2 0 0 0
% $\Omega_B (t)$
\put(19.3000,-32.4000){\makebox(0,0)[lb]{$\Omega_B (t)$}}%
% STR 2 0 3 0 Black White  
% 4 4492 3100 4492 3200 2 0 0 0
% $\Omega^\varepsilon_B (t)$
\put(44.9200,-32.0000){\makebox(0,0)[lb]{$\Omega^\varepsilon_B (t)$}}%
% SPLINE 2 0 3 0 Black White  
% 5 2590 1660 2630 1620 2730 1640 2730 1640 2730 1640
% 
\special{pn 8}%
\special{pa 2590 1660}%
\special{pa 2612 1634}%
\special{pa 2636 1618}%
\special{pa 2665 1616}%
\special{pa 2698 1626}%
\special{pa 2730 1640}%
\special{fp}%
% STR 2 0 3 0 Black White  
% 4 2740 1590 2740 1690 2 0 0 0
% $\partial \Omega$
\put(27.4000,-16.9000){\makebox(0,0)[lb]{$\partial \Omega$}}%
% STR 2 0 3 0 Black White  
% 4 5190 1530 5190 1630 2 0 0 0
% $\partial \Omega$
\put(51.9000,-16.3000){\makebox(0,0)[lb]{$\partial \Omega$}}%
% SPLINE 2 0 3 0 Black White  
% 5 5020 1610 5060 1570 5160 1590 5160 1590 5160 1590
% 
\special{pn 8}%
\special{pa 5020 1610}%
\special{pa 5042 1584}%
\special{pa 5066 1568}%
\special{pa 5095 1566}%
\special{pa 5128 1576}%
\special{pa 5160 1590}%
\special{fp}%
% VECTOR 2 0 3 0 Black White  
% 6 2140 1540 2140 1350 2140 1350 2140 1350 2140 1350 2140 1350
% 
\special{pn 8}%
\special{pa 2140 1540}%
\special{pa 2140 1350}%
\special{fp}%
\special{sh 1}%
\special{pa 2140 1350}%
\special{pa 2120 1417}%
\special{pa 2140 1403}%
\special{pa 2160 1417}%
\special{pa 2140 1350}%
\special{fp}%
\special{pa 2140 1350}%
\special{pa 2140 1350}%
\special{fp}%
\special{pa 2140 1350}%
\special{pa 2140 1350}%
\special{fp}%
% VECTOR 2 0 3 0 Black White  
% 6 4652 1530 4652 1340 4652 1340 4652 1340 4652 1340 4652 1340
% 
\special{pn 8}%
\special{pa 4652 1530}%
\special{pa 4652 1340}%
\special{fp}%
\special{sh 1}%
\special{pa 4652 1340}%
\special{pa 4632 1407}%
\special{pa 4652 1393}%
\special{pa 4672 1407}%
\special{pa 4652 1340}%
\special{fp}%
\special{pa 4652 1340}%
\special{pa 4652 1340}%
\special{fp}%
\special{pa 4652 1340}%
\special{pa 4652 1340}%
\special{fp}%
% STR 2 0 3 0 Black White  
% 4 2200 1400 2200 1500 2 0 0 0
% $n_\Omega$
\put(22.0000,-15.0000){\makebox(0,0)[lb]{$n_\Omega$}}%
% STR 2 0 3 0 Black White  
% 4 4682 1400 4682 1500 2 0 0 0
% $n_\Omega$
\put(46.8200,-15.0000){\makebox(0,0)[lb]{$n_\Omega$}}%
% SPLINE 2 0 3 0 Black White  
% 5 2550 2260 2590 2220 2690 2240 2690 2240 2690 2240
% 
\special{pn 8}%
\special{pa 2550 2260}%
\special{pa 2572 2234}%
\special{pa 2596 2218}%
\special{pa 2625 2216}%
\special{pa 2658 2226}%
\special{pa 2690 2240}%
\special{fp}%
% STR 2 0 3 0 Black White  
% 4 2720 2220 2720 2320 2 0 0 0
% $\Gamma (t)$
\put(27.2000,-23.2000){\makebox(0,0)[lb]{$\Gamma (t)$}}%
% SPLINE 2 0 3 0 Black White  
% 5 5100 2210 5140 2170 5240 2190 5240 2190 5240 2190
% 
\special{pn 8}%
\special{pa 5100 2210}%
\special{pa 5122 2184}%
\special{pa 5146 2168}%
\special{pa 5175 2166}%
\special{pa 5208 2176}%
\special{pa 5240 2190}%
\special{fp}%
% STR 2 0 3 0 Black White  
% 4 5250 2200 5250 2300 2 0 0 0
% $\Gamma^\varepsilon(t)$
\put(52.5000,-23.0000){\makebox(0,0)[lb]{$\Gamma^\varepsilon(t)$}}%
% VECTOR 2 0 3 0 Black White  
% 6 2100 2120 2100 1930 2100 1930 2100 1930 2100 1930 2100 1930
% 
\special{pn 8}%
\special{pa 2100 2120}%
\special{pa 2100 1930}%
\special{fp}%
\special{sh 1}%
\special{pa 2100 1930}%
\special{pa 2080 1997}%
\special{pa 2100 1983}%
\special{pa 2120 1997}%
\special{pa 2100 1930}%
\special{fp}%
\special{pa 2100 1930}%
\special{pa 2100 1930}%
\special{fp}%
\special{pa 2100 1930}%
\special{pa 2100 1930}%
\special{fp}%
% VECTOR 2 0 3 0 Black White  
% 6 4450 2120 4375 1945 4375 1945 4375 1945 4375 1945 4375 1945
% 
\special{pn 8}%
\special{pa 4450 2120}%
\special{pa 4375 1945}%
\special{fp}%
\special{sh 1}%
\special{pa 4375 1945}%
\special{pa 4383 2014}%
\special{pa 4396 1994}%
\special{pa 4420 1998}%
\special{pa 4375 1945}%
\special{fp}%
\special{pa 4375 1945}%
\special{pa 4375 1945}%
\special{fp}%
\special{pa 4375 1945}%
\special{pa 4375 1945}%
\special{fp}%
% STR 2 0 3 0 Black White  
% 4 4460 1950 4460 2050 2 0 0 0
% $n_{\Gamma^\varepsilon}$
\put(44.6000,-20.5000){\makebox(0,0)[lb]{$n_{\Gamma^\varepsilon}$}}%
\end{picture}}%